\documentclass[copyright,creativecommons]{eptcs}

\usepackage[utf8]{inputenc}
\usepackage[T1]{fontenc}
\usepackage{amssymb}
\usepackage{amsmath}
\usepackage{color}
\usepackage{xspace}
\usepackage{stmaryrd}
\usepackage{xcolor}
\usepackage{tikz}
\usepackage[caption=false]{subfig}
\usepackage{graphicx}
\usepackage{booktabs}
\usepackage{diagbox}
\usepackage{pifont}
\usepackage{colortbl}
\usepackage{adjustbox}
\usepackage{hyphenat}
\usepackage{amsthm}

\title{Compositional Liveness-Preserving Conformance Testing of Timed I/O Automata - Technical Report}
\author{Lars Luthmann\thanks{This work was funded by the Hessian LOEWE initiative within the Software-Factory 4.0 project.}
\institute{Real-Time Systems Lab\\TU Darmstadt, Germany}
\email{lars.luthmann@es.tu-darmstadt.de}
\and
Hendrik Göttmann
\institute{Real-Time Systems Lab\\TU Darmstadt, Germany}
\email{h.goettmann@stud.tu-darmstadt.de}
\and
Malte Lochau$^{*}$
\institute{Real-Time Systems Lab\\TU Darmstadt, Germany}
\email{malte.lochau@es.tu-darmstadt.de}
}
\def\titlerunning{Compositional Liveness-Preserving Conformance Testing of TIOA -- Technical Report}
\def\authorrunning{L. Luthmann, H. Göttmann, and M. Lochau}

\providecommand{\publicationstatus}{} 

\usetikzlibrary{calc,decorations.pathmorphing,shapes,arrows,positioning,automata}
\tikzset{
  every node/.style={},
  font=\footnotesize,
  initial text={},
  every initial by arrow/.style={*->},
  state/.style={circle,draw,fill=black,inner sep=.2em},
  emptystate/.style={inner sep=.2em},
  emptystateTA/.style={text width=3.2em,align=center},
  labeledcirclestate/.style={circle,draw,inner sep=.2em,align=center},
  labeledrectanglestate/.style={rectangle,draw,inner sep=.2em,align=center},
  labeledellipsestate/.style={ellipse,draw,inner sep=.2em,align=center},
  labeledroundedstate/.style={rounded corners,draw,inner sep=.2em,align=center},
  labeledroundedstateTA/.style={rounded corners,draw,text width=4.1em,align=center},
  labeledroundedstateTAsmall/.style={rounded corners,draw,text width=1.5em,align=center},
  labeledstate/.style={inner sep=.2em,align=center},
  safe/.style={->,>=stealth',dashed},
  live/.style={->,>=stealth'},
  empty/.style={->,>=stealth',transparent}
}

\newcommand{\scalefactor}{.85} 

\newcommand{\ie}{\mbox{i.e.,}\xspace}
\newcommand{\eg}{\mbox{e.g.,}\xspace}
\newcommand{\st}{\mbox{s.t.}\xspace}
\newcommand{\Wlog}{\mbox{w.l.o.g.}\xspace}
\newcommand{\etal}{\mbox{et al.}\xspace}
\newcommand{\cf}{\mbox{cf.}\xspace}
\newcommand{\angles}[1]{\ensuremath{\left\langle#1\right\rangle}}

\newcommand{\cmark}{\ding{51}}
\newcommand{\xmark}{\ding{55}}
\newcommand{\specialcell}[2][c]{\begin{tabular}[#1]{@{}#1@{}}#2\end{tabular}}

\renewcommand{\xrightarrow}[1]{
  \mathrel{
    \!\!
    \tikz[baseline=-\the\dimexpr\fontdimen22\textfont2\relax]{
      \node[anchor=south,font=\scriptsize, inner ysep=1.5pt,outer xsep=2.2pt](x){\ensuremath{\!#1}};
      \draw[shorten <=3.4pt,shorten >=3.4pt,->](x.south west)--(x.south east);
    }
    \!\!
  }
}

\newcommand{\xtwoheadrightarrow}[1]{
  \mathrel{
    \!\!
    \tikz[baseline=-\the\dimexpr\fontdimen22\textfont2\relax]{
      \node[anchor=south,font=\scriptsize, inner ysep=1.5pt,outer xsep=2.2pt](x){\ensuremath{\!#1}};
      \draw[shorten <=3.4pt,shorten >=3.4pt,->>](x.south west)--(x.south east);
    }
    \!\!
  }
}

\newcommand{\xTwoheadrightarrow}[1]{
  \mathrel{
    \begin{tikzpicture}[baseline= {( $ (current bounding box.south) + (0,-0.5ex) $ )}]
        \node[inner sep=.5ex] (a) {$\!\scriptstyle #1\,$};
        \path[draw,implies-,double distance between line centers=2pt] (a.south east) -- (a.south west);
        \path[draw,implies-,double distance between line centers=2pt] ($ (a.south east) + (-0.079,0) $)--($ (a.south east) + (-0.08,0) $);
    \end{tikzpicture}
  }
}

\newcommand{\xrightsquigarrow}[1]{
  \mathrel{
    \begin{tikzpicture}[baseline= {( $ (current bounding box.south) + (0,-0.5ex) $ )}]
      \node[inner sep=.5ex] (a) {\ensuremath{\scriptstyle#1}};
      \path[draw,<-,decorate,decoration={zigzag,amplitude=.9pt,segment length=1.2mm,pre=lineto,pre length=3pt}] (a.south east) -- (a.south west);
    \end{tikzpicture}
  }
}

\newcommand{\xRightsquigarrow}[1]{
  \mathrel{
    \begin{tikzpicture}[baseline= {( $ (current bounding box.south) + (0,-0.5ex) $ )}]
      \node[inner sep=.5ex] (a) {$\scriptstyle #1$};
      \path[draw,implies-,double distance between line centers=1.5pt,decorate,decoration={zigzag,amplitude=0.7pt,segment length=1.2mm,pre=lineto,pre length=3pt}] (a.south east) -- (a.south west);
    \end{tikzpicture}
  }
}

\newcommand{\Twoheadrightarrow}{\xTwoheadrightarrow{\hspace{.45em}}}
\newcommand{\Rightsquigarrow}{\xRightsquigarrow{\hspace{.4em}}}

\newcommand{\traces}{\ensuremath{\textit{traces}}\xspace}
\newcommand{\ttraces}{\ensuremath{\textit{ttraces}}\xspace}
\newcommand{\sptraces}{\ensuremath{\textit{sptraces}}\xspace}
\newcommand{\tstraces}{\ensuremath{\textit{tstraces}}\xspace}
\newcommand{\tstracesl}{\ensuremath{\tstraces_L}\xspace}

\DeclareMathOperator{\after}{\mathbf{after}}
\DeclareMathOperator{\out}{\mathbf{out}}
\DeclareMathOperator{\mathtiocodelta}{\mathbf{tioco}_\delta}
\DeclareMathOperator{\mathtiocodelay}{\mathbf{tioco}_\Delta}
\DeclareMathOperator{\mathltiocolts}{\mathbf{ltioco}_{\textit{S}}}
\DeclareMathOperator{\mathltiocozg}{\mathbf{ltioco}_\mathcal{Z}}
\DeclareMathOperator{\afterz}{\after_\mathcal{Z}}
\DeclareMathOperator{\outz}{\out_\mathcal{Z}}
\DeclareMathOperator{\outset}{\widehat{\out}_\mathcal{Z}}
\DeclareMathOperator{\outdelay}{\out_\Delta}
\DeclareMathOperator{\outl}{\out_S}

\newcommand{\subsetsim}{\mathrel{\substack{\textstyle\subset\\[-0.2ex]\textstyle\sim}}}

\newtheorem{definition}{Definition}
\newtheorem{theorem}{Theorem}
\newtheorem{lemma}{Lemma}
\newtheorem{corollary}{Corollary}
\newtheorem{proposition}{Propostion}
\newtheorem{example}{Example}

\begin{document}

\maketitle

%
\begin{abstract}
I/O conformance testing theories (\eg{} \emph{ioco}) are concerned with formally defining when observable output behaviors of an implementation conform to those permitted by a specification.
Thereupon, several real-time extensions of \emph{ioco}, usually called \emph{tioco}, have been proposed, further taking into account permitted delays between actions.
In this paper, we propose an improved version of \emph{tioco}, called \emph{live timed ioco} (\emph{ltioco}), tackling various weaknesses of existing definitions.
Here, a reasonable adaptation of quiescence (\ie{} observable absence of any outputs) to real-time behaviors has to be done with care: \emph{ltioco} therefore distinguishes safe outputs being allowed to happen, from live outputs being enforced to happen within a certain time period thus inducing two different facets of quiescence.
Furthermore, \emph{tioco} is frequently defined on Timed I/O Labeled Transition Systems (TIOLTS), a semantic model of Timed I/O Automata (TIOA) which is infinitely branching and thus infeasible for practical testing tools.
Instead, we extend the theory of zone graphs to enable \emph{ltioco} testing on a finite semantic model of TIOA.
Finally, we investigate compositionality of \emph{ltioco} with respect to parallel composition including a proper treatment of silent transitions.
\end{abstract}
%
\section{Introduction}\label{sec:introduction}
Model-based testing constitutes a practically emerging, yet theoretically founded technique
for automated quality assurance of software systems~\cite{Broy2005}.
In particular, input/output conformance testing theories formalize
notions of observable conformance between an implementation under test 
and a specification, where the \textbf{ioco} theory~\cite{Tretmans1996} 
constitutes one of the most prominent examples.
The \textbf{ioco} relation requires both the input/output-behaviors 
of the specification and the implementation to be represented as
input/output labeled transition systems (IOLTS), where the IOLTS
of the implementation is unknown (black-box assumption)~\cite{Bernot1991}.
For an implementation to satisfy \textbf{ioco}, all its
possible output behaviors must be permitted by the specification.
To rule out trivial implementations never showing any output, 
\textbf{ioco} employs the notion of \emph{quiescence} to explicitly permit starvation.
In order to ensure proper test-execution semantics, \textbf{ioco}
requires \emph{input-enabled} implementations, never blocking any \mbox{(test-)}inputs.
Hence, \textbf{ioco} is concerned with the correct
\emph{ordering} of (or causality among) input/output (re-)actions, whereas
quantified \emph{time delays} between action occurrences are not considered.
However, reasoning about real-time behaviors becomes more and more crucial 
and various real-time extensions of \textbf{ioco}, so-called \textbf{tioco},
have been recently proposed~\cite{Schmaltz2008,BrandanBriones2004,Hessel2008,Krichen2004,Larsen2004}.
Based on timed extensions of IOLTS (so-called TIOLTS), 
a system run progresses by either actively performing discrete, instantaneous actions 
or by inactively letting a quantified amount of time pass.
Nevertheless, existing definitions of \textbf{tioco} suffer from several weaknesses
which we tackle in this paper by proposing an 
improved version called \emph{live timed ioco} \textbf{(ltioco)}.
Our contributions can be summarized as follows.

\begin{itemize}
  \item Recent adoptions of quiescence in a timed setting also show several weaknesses: 
  most recent versions of \textbf{tioco} either do not incorporate any notion of 
  quiescence at all~\cite{Schmaltz2008,Hessel2008,Krichen2004,Larsen2005}, or define quiescence in terms of (either infinite or bounded) time intervals 
without observable output actions~\cite{BrandanBriones2004,Schmaltz2008}.
Both fail to distinguish the \emph{enabling} of output actions (\ie{} an output 
is allowed to occur in a time interval to constitute \emph{safe} behavior) from
\emph{enforced} output actions (\ie{} an output must occur in a 
certain time interval to meet \emph{liveness} requirements).
To this end, \textbf{ltioco} distinguishes
safe outputs from live outputs thus explicitly
incorporating the two different facets of timed quiescence.
We prove correctness of \textbf{ltioco} with respect to TIOLTS semantics and we show that
\textbf{ltioco} is strictly more discriminating 
than most recent versions of \textbf{tioco}.

  \item We investigate compositionality properties of \textbf{ltioco} with 
respect to (synchronous) parallel composition including silent transitions.
  
  \item Finally, all recent versions of \emph{tioco} are defined 
on TIOLTS, constituting a semantic model of Timed I/O Automata (TIOA)
which is infinitely branching and thus infeasible for practical testing tools.
Instead, we extend the notion of zone graphs to effectively check \emph{ltioco} on a finite semantic model of TIOA
using so-called \emph{span traces}.
Thereupon, we developed a tool for online testing using \textbf{tioco}
(see \url{https://www.es.tu-darmstadt.de/ltioco}).
\end{itemize}

The remainder of this paper is structured as follows.
We first give an formal introduction into TIOA and parallel composition of TIOA in~Sect.~\ref{sec:preliminaries}.
Then, we discuss existing notions of \textbf{tioco} and point out their weaknesses in~Sect.~\ref{sec:tioco} 
which we address in the subsequent Sect.~\ref{sec:ltioco}.
Furthermore, we give an intuition on how to apply zone graphs for an efficient implementation 
of our approach in Sect.~\ref{sec:implementation} and we summarize related work in~Sect.~\ref{sec:related-work}.
%
\section{Timed Input/Output Automata}\label{sec:preliminaries}
We first recall foundations of \emph{Timed Automata (TA)}~\cite{Alur1990,Alur1994},
extension of TA by input/output labels~\cite{Lynch1992,Merritt1991,David2010} 
and their composition involving silent transitions~\cite{Berard1998}.

TA are labeled finite state-transition graphs with
states being called \emph{locations} and transitions being called \emph{switches}.
A TA is further defined with respect to a finite set $\mathcal{C}$ of \textit{clocks} 
over a numerical \emph{clock domain} $\mathbb{T}$ (\eg{} $\mathbb{T}=\mathbb{N}_0$ 
for \emph{discrete time} and $\mathbb{T}=\mathbb{R}_+$ 
with $\mathbb{R}_+:=\{r\mid r\in\mathbb{R}\land r\geq 0\}$ for \emph{dense time}).
Clocks constitute constantly and synchronously increasing, yet independently resettable
variables over $\mathbb{T}$ for measuring and restricting
time intervals (durations/delays) between action occurrences.
Note that we consider $\mathbb{T}=\mathbb{N}_0$ in all examples for the sake of readability.
In particular, we consider \emph{Timed Safety Automata}~\cite{Henzinger1994} in which
time-critical behaviors are expressed by \emph{clock constraints} as
\emph{guards} for switches and \emph{invariants} for locations.
Guards restrict time intervals in which a switch is enabled while residing
in its source location, whereas invariants restrict time intervals in which a TA run 
is permitted to reside in a location.
Alternative TA definitions may incorporate distinguished \emph{acceptance locations} thus
employing Büchi acceptance semantics on \emph{infinite runs}~\cite{Alur1990,Henzinger1994}
which is out of the scope of this paper as model-based testing is 
inherently limited to \emph{finite} test runs.

\emph{Timed Input/Output-labeled Automata (TIOA)}
extend TA for timed interface specifications (\eg{} 
for model-based conformance testing of time-critical components or systems~\cite{Lynch1992,Merritt1991}).
The \emph{label alphabet} $\Sigma=\Sigma_I\cup \Sigma_O$ of a TIOA
consists of two disjoint subsets of 
(externally controllable, internally observable) \emph{input actions} $\Sigma_I$ and 
(externally observable, internally controllable) \emph{output actions} $\Sigma_O$.
The special symbol $\tau\notin\Sigma$ summarizes \emph{internal actions} 
of silent switches being neither externally controllable 
nor visible, and we write $\Sigma_{\tau}=\Sigma\cup\{\tau\}$ for short.

\begin{definition}[TIOA]\label{def:tioa}
A \emph{TIOA} $\mathcal{A}$ is a tuple $\left(L,\ell_0,\Sigma_I,\Sigma_O,\rightarrow,I\right)$, where
\begin{itemize}
  \item $L$ is a finite set of \emph{locations} with \emph{initial location} $\ell_0\in L$,
  \item $\Sigma_I$ and $\Sigma_O$ are sets of \emph{input actions} and \emph{output actions} with $\Sigma_I\cap \Sigma_O=\emptyset$,
  \item ${\rightarrow}\subseteq L\times\mathcal{B(C)}\times\Sigma_{\tau}\times 2^\mathcal{C}\times L$ is a relation defining \emph{switches}, 
  with a set $\mathcal{B(C)}$ of \emph{clock constraints} $\varphi$ inductively defined as
  $$\varphi:=x\sim r\mid x-y\sim r\mid\neg\varphi\mid\varphi\land\varphi\mid\mathsf{true},$$
  where $x,y\in\mathcal{C}$, $r\in\mathbb{Q}_+$, and ${\sim}\in\{<,\leq,=,\geq,>\}$, and
  \item $I:L\rightarrow\mathcal{B(C)}$ is a function assigning \emph{location invariants}.
\end{itemize}
\end{definition}
We write $\ell \xrightarrow{g,\sigma,R} \ell'$ to denote
switches from location $\ell$ to $\ell'$ with guard $g$, 
action $\sigma$ and set $R\subseteq \mathcal{C}$ of clocks being reset.
Without loss of generality, we assume each location invariant being
unequal to $\mathsf{true}$ to be 
\emph{downward-closed} (\ie{} with clauses $x\leq r$ or $x<r$)~\cite{Bengtsson2004}.
The operational semantics of TIOA may be defined
as \emph{Timed Input/Output Labeled Transition System} (\emph{TIOLTS})~\cite{Henzinger1991}.
A TIOLTS state $\angles{\ell,u}$ is a pair consisting of a location $\ell\in L$ and a
\emph{clock valuation} $u\in\mathcal{C}\rightarrow\mathbb{T}$.
A TIOLTS defines two kinds of transitions:
(1) passage of time while inactively residing in a location, and
(2) instantaneous switches between locations due to action occurrences (including $\tau$). 
Given a clock valuation $u$, $u+d$ denotes the clock valuation mapping each clock $c\in\mathcal{C}$ 
to the updated clock value $u(c)+d$ with $d\in\mathbb{T}$.
For a subset $R\subseteq\mathcal{C}$ of clocks, $[R\mapsto 0]u$ denotes 
the clock valuation mapping every clock in $R$ to 0 while 
preserving the values of all other clocks in $\mathcal{C}\setminus R$.
Finally, $u\in g$ denotes that clock valuation $u$ satisfies
clock constraint $g\in\mathcal{B(C)}$.
We further distinguish between
\emph{strong} and \emph{weak} transitions, depending on whether
silent transitions are visible or not.

\begin{definition}[TIOLTS]\label{def:tiolts-semantics}
The TIOLTS of TIOA $(L,\ell_0,\Sigma_I,\Sigma_O,\rightarrow,I)$ 
is a tuple $(S,s_0,\Sigma_I,\Sigma_O,\twoheadrightarrow)$, where 
\begin{itemize}
  \item $S=L\times(\mathcal{C}\rightarrow\mathbb{T})$ is a set of \emph{states}
  with \emph{initial state} $s_0=\angles{\ell_0,[\mathcal{C}\mapsto 0]u_0}\in S$,
  \item $\hat\Sigma_{\tau}=\Sigma_{I}\cup\Sigma_{O}\cup\{\tau\}\cup\Delta$ is a 
set of \emph{labels} with $\Delta=\mathbb{T}$, $\Sigma_{\tau}\cap\Delta=\emptyset$, and
  \item ${\twoheadrightarrow}\subseteq S\times\hat\Sigma_{\tau}\times S$ is a set of \emph{(strong) transitions} 
being the least relation satisfying the rules:
\begin{itemize}
  \item $\angles{\ell,u}\xtwoheadrightarrow{d}\angles{\ell,u+d}$ if $u\in I(\ell)$ and $(u+d)\in I(\ell)$ for $d\in\mathbb{T}$, and
  \item $\angles{\ell,u}\xtwoheadrightarrow{\sigma}\angles{\ell',u'}$ if $\ell\xrightarrow{g,\sigma,R}\ell'$, $u\in g$, $u'=[R\mapsto 0]u$, $u'\in I(\ell')$, $\sigma\in\Sigma_{\tau}$.
\end{itemize}
\end{itemize}
By ${\Twoheadrightarrow}\subseteq S\times \hat{\Sigma}\times S$ we further
denote a set of \emph{(weak) transitions} being the least relation satisfying the rules:
  \begin{itemize}
  \item $s_0\xtwoheadrightarrow{\tau^n}s_n$ if $\exists s_1,\ldots s_{n-1}\in S:s_0\xtwoheadrightarrow{\tau}s_1\xtwoheadrightarrow{\tau}\ldots\xtwoheadrightarrow{\tau}s_n$ with $n\in \mathbb{N}_0$,
  \item $s\xTwoheadrightarrow{\sigma}s'$ if $\exists s_1,s_2\in S:s\xtwoheadrightarrow{\tau^{n}}s_1\xtwoheadrightarrow{\sigma}s_2\xtwoheadrightarrow{\tau^{m}}s'$ with $n, m\in \mathbb{N}_{0}$,
  \item $s\xTwoheadrightarrow{d}s'$ if $s\xtwoheadrightarrow{d}s'$,
  \item $s\xTwoheadrightarrow{0}s'$ if $s\xtwoheadrightarrow{\tau^{n}}s'$ with $n\in \mathbb{N}_{0}$,
  \item $s_0\xTwoheadrightarrow{\sigma_1\cdots\sigma_n}$ if $\exists s_1,\ldots s_n\in S:s_0\xTwoheadrightarrow{\sigma_1}s_1\xTwoheadrightarrow{\sigma_2}\ldots\xTwoheadrightarrow{\sigma_n}s_n$ with $n\in \mathbb{N}_0$, and
  \item $s\xTwoheadrightarrow{d+d'}s'$ if $\exists s''\in S:s\xTwoheadrightarrow{d}s''$ and $s''\xTwoheadrightarrow{d'}s'$.
\end{itemize}
\end{definition}
We only consider \emph{strongly convergent} TIOA (\ie{} having TIOLTS without infinite $\tau$-sequences).
By $\llbracket\mathcal{A}\rrbracket_{S}^{x}$, $x\in\{w,s\}$, we refer to the (either weak or strong)
TIOLTS semantics of TIOA $\mathcal{A}$, where we omit parameter $x$ if not relevant.
The weak semantics is obtained by replacing all occurrences of $\twoheadrightarrow$ by $\Twoheadrightarrow$ 
in all definitions.
We recall three essential 
properties for strong 
TIOLTS semantics of any given TIOA~\cite{David2010,Aceto1998}.

\begin{proposition}\label{proposition:tiolts-properties}
    Let $(S,s_0,\Sigma_I,\Sigma_O,\twoheadrightarrow)$ be a TIOLTS of a TIOA.
    \begin{itemize}
      \item (Time Add) $\forall s_1,s_3\in S,\forall d_1,d_2\in\Delta:s_1\xtwoheadrightarrow{d_1+d_2}s_3 \Leftrightarrow \exists s_2:s_1\xtwoheadrightarrow{d_1}s_2\xtwoheadrightarrow{d_2}s_3$
      \item (Time Reflex) $\forall s_1,s_2\in S:s_1\xtwoheadrightarrow{0}s_2\Rightarrow s_1=s_2$
      \item (Time Determ) $\forall s_1,s_2,s_3\in S:s_1\xtwoheadrightarrow{d}s_2$ and $s_1\xtwoheadrightarrow{d}s_3$ then $s_2=s_3$
    \end{itemize}
\end{proposition}
In contrast, the weak semantics obviously obstructs all three properties.

Furthermore, by $\traces(s_0)=\{\omega \mid s_0\xtwoheadrightarrow{\omega}\}$
we denote the set of all \emph{traces} $\omega=\alpha_{1}\alpha_{2}\cdots\alpha_{k}\in(\Sigma\cup\Delta)^{*}$
corresponding to some path $s_{0}\xtwoheadrightarrow{\alpha_{1}}s_{1}\xtwoheadrightarrow{\alpha_{2}}\cdots\xtwoheadrightarrow{\alpha_{k}}s_{k}$
of TIOLTS $s$.
Given a TIOA $\mathcal{A}$, the TIOLTS $\llbracket\mathcal{A}\rrbracket_{S}$
defines all possible \emph{(timed) runs} 
$s_0=\langle \ell_0, u_0\rangle\xtwoheadrightarrow{d_1}\xtwoheadrightarrow{\sigma_1}\langle \ell_1,u_1\rangle\xtwoheadrightarrow{d_2}\xtwoheadrightarrow{\sigma_2}\cdots$
of $\mathcal{A}$ in terms of sequences of \emph{(timed) steps}
$s\xtwoheadrightarrow{d}\xtwoheadrightarrow{\sigma}s''$ denoting
$\exists s'\in S: s\xtwoheadrightarrow{d}s'\xtwoheadrightarrow{\sigma}s''$~\cite{Schmaltz2008}.
We refer to the set of weak/strong traces of state $s$ by $\traces(s)^x$, $x\in\{w,s\}$, respectively.

\begin{example}
Figure~\ref{fig:vending-machine-ta} shows a (simplified) TIOA $\mathcal{A}_{1}$ of
a vending machine with two clocks, $x$ and $y$, and Fig.~\ref{fig:vending-machine-tiolts} depicts an extract from its TIOLTS.
\begin{figure}[tp]
  \hfill
  \subfloat[TIOA $\mathcal{A}_{1}$]{\label{fig:vending-machine-ta}

\scalebox{\scalefactor}{
\begin{tikzpicture}[node distance=1.2]

\node[labeledroundedstateTA, initial, initial where=below] (idle) {idle\\$x\leq20$};
\node[labeledroundedstateTA, left=of idle] (off) {off};
\node[labeledroundedstateTA, right=of idle] (as) {add\\sugar};
\node[labeledroundedstateTA, above=of as] (pc) {preparing\\coffee\\$y\leq20$};
\node[labeledroundedstateTA, left=of pc] (done) {done\\$y\leq20$};

\draw[live] (idle) to node[pos=.5, align=center] {$\tau$\\$[x=20]$} (off);
\draw[live] (idle) to node[pos=.5, align=center] {\\\\?press\\$x:=0$\\$y:=0$} (as);
\draw[live,loop below] (as) to node[pos=.43, right, align=left] {?sugar\\$[x\geq10]$\\$x:=0$} (as);
\draw[live] (as) to node[pos=.5, right, align=left] {!proceed\\$[y\leq20]$\\$y:=0$} (pc);
\draw[live, bend right=25] (pc) to node[pos=.5, align=center] {!coffee\\$[y>15]$} (done);
\draw[live, bend left=25] (pc) to node[pos=.5, align=center] {!coffee\\$[y\leq15]$} (done);
\draw[live] (done) to node[auto, left, align=right] {$\tau$\\$x:=0$\\$y:=0$} (idle);
\draw[live, bend right=45] (off) to node[pos=.53, align=center] {\\\\?press\\$x:=0$\\$y:=0$} (idle);

\end{tikzpicture}
}}
  \hfill
  \subfloat[TIOLTS $\llbracket\mathcal{A}_{1}\rrbracket_S$]{\label{fig:vending-machine-tiolts}

\scalebox{\scalefactor}{
\begin{tikzpicture}[node distance=.5]

\node[labeledstate, initial, initial where=left] (idle0) {$\angles{\text{idle},x=0,y=0}$};
\node[labeledstate, below=of idle0] (idle1) {$\angles{\text{idle},x=1,y=1}$};
\node[labeledstate, below=of idle1] (idle1dots) {\vdots};
\node[emptystate, below=of idle1] (idle1dotsempty) {\phantom{$\angles{\text{idle},x=1,y=1}$}};
\node[labeledstate, below=of idle1dotsempty] (idle20) {$\angles{\text{idle},x=20,y=20}$};
\node[labeledstate, below=of idle20] (off20) {$\angles{\text{off},x=20,y=20}$};
\node[labeledstate, below=of off20] (off21) {};
\node[emptystate, below right=of off20] (off20empty) {};

\node[labeledstate, right=of idle20] (as1020) {$\angles{\text{as},x=10,y=20}$};
\node[labeledstate, above=of as1020] (as010) {$\angles{\text{as},x=0,y=10}$};
\node[labeledstate, above=of as010] (as100) {$\angles{\text{as},x=10,y=10}$};
\node[labeledstate, above=of as100] (as00) {$\angles{\text{as},x=0,y=0}$};
\node[labeledstate, below=of as1020] (as020) {$\angles{\text{as},x=0,y=20}$};
\node[labeledstate, below=of as020] (as020dots) {};

\node[labeledstate, right=of as100] (as100right) {};
\node[labeledstate, right=of as1020] (as1020right) {};
\node[labeledstate, right=of as020] (as020right) {};
\node[labeledstate, right=of as00] (as00right) {};
\node[labeledstate, right=of as010] (as010right) {};

\draw[live] (idle0) to node[auto,swap] {$1$} (idle1);
\draw[live] (idle1) to node[auto,swap] {$1$} (idle1dots);
\draw[live] (idle1dots) to node[auto,swap] {$1$} (idle20);
\draw[live] (idle20) to node[auto] {$\tau$} (off20);
\draw[live] (off20) to node[auto] {$1$} (off21);
\draw[live, bend left=72] (off20) to node[auto,swap] {?press} (idle0);

\draw[live] (as00) to node[auto] {$10$} (as100);
\draw[live] (as100) to node[auto] {?sugar} (as010);
\draw[live] (as010) to node[auto] {$10$} (as1020);
\draw[live] (as1020) to node[auto] {?sugar} (as020);
\draw[live] (as020) to node[auto] {$1$} (as020dots);

\draw[live] (as1020) to node[auto] {!proc.} (as1020right);
\draw[live] (as020) to node[auto] {!proc.} (as020right);
\draw[live] (as100) to node[auto] {!proc.} (as100right);
\draw[live] (as00) to node[auto] {!proc.} (as00right);
\draw[live] (as010) to node[auto] {!proc.} (as010right);

\draw[live] (idle0) to node[auto] {?press} (as00);
\draw[live, bend left=5] (idle1) to node[pos=.5, left] {?press} (as00);
\draw[live, bend left=10] (idle20) to node[pos=.4, left] {?press} (as00);

\end{tikzpicture}
}}
  \hfill\strut
  \caption{TIOA for a Simple Vending Machine~\cite{Andre2016,Bengtsson2004} and Extract from TIOLTS}
\end{figure}
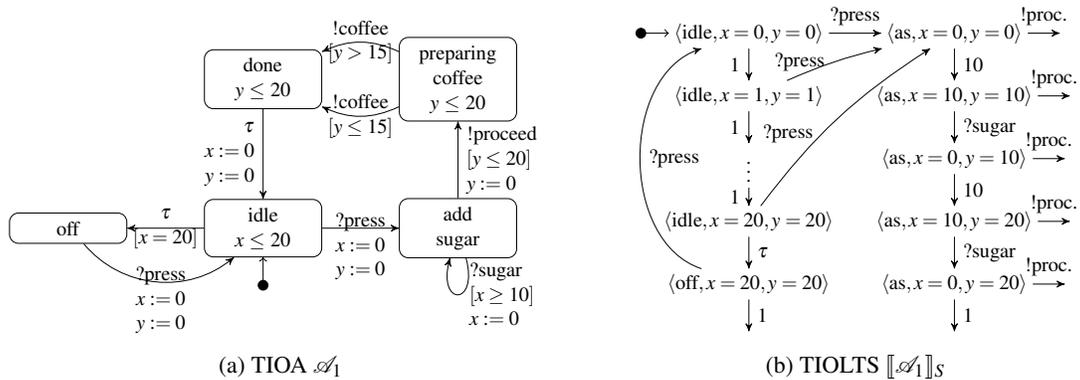
Switches are labeled with actions (prefixes ``?'' for inputs and ``!'' for outputs), 
guards (\eg{} $x\leq 20$), and (possibly empty) clock resets.
We label locations by their names (\eg{} initial location \emph{idle}) and their location invariants.
Clock constraints being equal to $\mathsf{true}$ are omitted.
Each (timed) \emph{run} of the machine starts in initial location \emph{idle}, where 
a user may \emph{press} a button to switch to location \emph{add sugar}.
If no button is pressed for 20 time units (\eg{} seconds), 
the machine is turned \emph{off} via a silent switch and may be 
switched to \emph{idle}, again, by pressing a button.
In location \emph{add sugar}, sugar may be repeatedly selected, 
where at least 10 seconds must pass between two consecutive requests and
the machine proceeds to location \emph{preparing coffee} at most 20 second after input \emph{press}.
Here, coffee is dispensed for at most 20 seconds and 
the machine finally returns to \emph{idle}.
The machine either produces small coffees 
(finishing after less than 15 seconds) or
large coffees (requiring more than 15 seconds).
This example illustrates the semantic differences between 
guards and invariants: guards restrict time intervals in which
a switch is \emph{allowed} to be taken, whereas
invariants define time intervals after which
a location is \emph{enforced} to be left 
(\eg{} it is allowed to perform \textit{!proceed} 
to leave location \emph{add sugar}
while $y\leq20$ holds, whereas it is enforced to leave 
location \emph{preparing coffee} in case of $y=20$).
Hence, guards express \emph{safety} conditions, whereas invariants express
\emph{liveness} conditions of timed runs.
\end{example}
A TIOA is supposed to specify one particular part 
of an arbitrary complex system composed of several concurrently interacting \emph{components}.
We define CCS-like \emph{parallel composition} of TIOA with synchronous
communication via shared input/output actions, becoming internal $\tau$-actions~\cite{David2010}.
As a prerequisite for composing two TIOA $\mathcal{A}_{1}$ and $\mathcal{A}_{2}$, 
denoted as $\mathcal{A}_{1\parallel 2}=\mathcal{A}_{1}\parallel \mathcal{A}_{2}$,
we require both to be \emph{composable} (\ie{} all shared actions have opposed directions).

\begin{definition}[TIOA Composition]\label{def:tioa-compo-epsi}
Let $\left(L_j,\ell_{0_j},\Sigma_{I_j},\Sigma_{O_j},\rightarrow_j, I_j \right)$ with $j\in\{1,2\}$ 
be TIOA with $\Sigma_{I_1} \cap \Sigma_{I_2} = \emptyset$, $\Sigma_{O_1} \cap \Sigma_{O_2} = \emptyset$ and $\mathcal{C}_{1}\cap\mathcal{C}_{2} = \emptyset$.
Their \emph{parallel composition} is a TIOA
$(L_1\times L_2,(\ell_{0_1},\ell_{0_2}),\Sigma_{I_{1\parallel 2}},\Sigma_{O_{1\parallel 2}},\rightarrow_{1 \parallel 2},I_{1 \parallel 2})$
over $\mathcal{C}_{1\parallel 2}=\mathcal{C}_{1} \cup \mathcal{C}_{2}$ with
  $\Sigma_{I_{1\parallel 2}}= (\Sigma_{I_{1}}\cup \Sigma_{I_{2}})\setminus (\Sigma_{O_{1}}\cup \Sigma_{O_{2}})$, 
  $\Sigma_{O_{1\parallel 2}}= (\Sigma_{O_{1}}\cup \Sigma_{O_{2}})\setminus (\Sigma_{I_{1}}\cup \Sigma_{I_{2}})$, 
  $I_{1 \parallel 2}(\ell_1,\ell_2)=I_{1}(\ell_1)\wedge I_{2}(\ell_2)$, and
  $\rightarrow_{1 \parallel 2}$ is the least relation satisfying the rules:
\begin{tabbing}
  (3) \= $(\ell_1,\ell_2)\xrightarrow{g_1\wedge g_2,\tau,R_{1}\cup R_{2}}_{1\parallel 2} (\ell_1',\ell_2')$ \= if \= \kill
  (1) \> $(\ell_1,\ell_2)\xrightarrow{g_1, \sigma, R_{1}}_{1\parallel 2} (\ell_1',\ell_2)$ \> if \> $\ell_1\xrightarrow{g_1,\sigma, R_{1}}_{1} \ell_1'$ and $\sigma\in(\Sigma_1\setminus\Sigma_2)\cup\{\tau\}$ \\
  (2) \> $(\ell_1,\ell_2)\xrightarrow{g_2, \sigma, R_{2}}_{1\parallel 2} (\ell_1,\ell_2')$ \> if \> $\ell_2\xrightarrow{g_2,\sigma, R_{2}}_{2} \ell_2'$ and $\sigma\in(\Sigma_2\setminus\Sigma_1)\cup\{\tau\}$ \\
  (3) \> $(\ell_1,\ell_2)\xrightarrow{g_1\wedge g_2,\tau,R_{1}\cup R_{2}}_{1\parallel 2} (\ell_1',\ell_2')$ \> if \> $\ell_1\xrightarrow{g_{1},\sigma,R_{1}}_{1} \ell_1'$, $\ell_2'\xrightarrow{g_{2},\sigma,R_{2}}_{2} \ell_2'$ and \\
  \> \> \> $\sigma\in(\Sigma_1\cap\Sigma_2)$. \\
\end{tabbing}
\vspace{-.75cm}
\end{definition}

\begin{example}
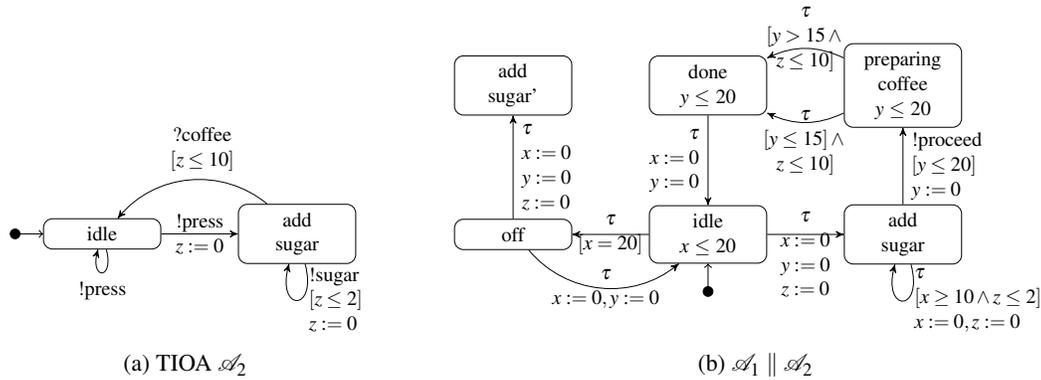
\begin{figure}[tp]
  \hfill
  \subfloat[TIOA $\mathcal{A}_{2}$]{\label{fig:tioa-composition-q}

\scalebox{\scalefactor}{
\begin{tikzpicture}[node distance=1.2]

\node[labeledroundedstateTA, initial, initial where=left] (idle) {idle};
\node[labeledroundedstateTA, right=of idle] (as) {add\\sugar};

\draw[live] (idle) to node[pos=.5, align=center] {!press\\$z:=0$} (as);
\draw[live,loop below] (idle) to node[auto] {!press} (idle);
\draw[live,loop below] (as) to node[pos=.43, right, align=left] {!sugar\\$[z\leq2]$\\$z:=0$} (as);
\draw[live, bend right=45] (as) to node[pos=.42,above,align=center] {?coffee\\$[z\leq10]$} (idle);

\end{tikzpicture}
}}
  \hfill
  \subfloat[$\mathcal{A}_{1}\parallel \mathcal{A}_{2}$]{\label{fig:tioa-composition-result}

\scalebox{\scalefactor}{
\begin{tikzpicture}[node distance=1.2]

\node[labeledroundedstateTA, initial, initial where=below] (idle) {idle\\$x\leq20$};
\node[labeledroundedstateTA, left=of idle] (off) {off};
\node[labeledroundedstateTA, right=of idle] (as) {add\\sugar};
\node[labeledroundedstateTA, above=of as] (pc) {preparing\\coffee\\$y\leq20$};
\node[labeledroundedstateTA, left=of pc] (done) {done\\$y\leq20$};
\node[labeledroundedstateTA, left=of done] (asp) {add\\sugar'};

\draw[live] (idle) to node[pos=.5, align=center] {$\tau$\\$[x=20]$} (off);
\draw[live] (idle) to node[pos=.5, align=center] {\\\\$\tau$\\$x:=0$\\$y:=0$\\$z:=0$} (as);
\draw[live,loop below] (as) to node[pos=.43, right, align=left] {$\tau$\\$[x\geq10\land z\leq2]$\\$x:=0,z:=0$} (as);
\draw[live] (as) to node[pos=.5, right, align=left] {!proceed\\$[y\leq20]$\\$y:=0$} (pc);
\draw[live, bend right=25] (pc) to node[pos=.5, align=center] {$\tau$\\$[y>15\land\mathstrut$\\$z\leq10]$\\} (done);
\draw[live, bend left=25] (pc) to node[pos=.5, align=center] {\\\\$\tau$\\$[y\leq15]\land\mathstrut$\\$z\leq10]$} (done);
\draw[live] (done) to node[auto, left, align=right] {$\tau$\\$x:=0$\\$y:=0$} (idle);
\draw[live] (off) to node[right, pos=.5, align=left] {$\tau$\\$x:=0$\\$y:=0$\\$z:=0$} (asp);
\draw[live, bend right=45] (off) to node[pos=.53, align=center] {$\tau$\\$x:=0,y:=0$} (idle);

\end{tikzpicture}
}}
  \hfill\strut
  \caption{Sample TIOA Composition}\label{fig:tioa-composition}
\end{figure}
Consider TIOA $\mathcal{A}_{1}$, $\mathcal{A}_{2}$ and their parallel 
composition $\mathcal{A}_{1}\parallel\mathcal{A}_{2}$ 
(\cf{}~Figs.~\ref{fig:vending-machine-ta}, \ref{fig:tioa-composition-q}, and~\ref{fig:tioa-composition-result}).
A customer $\mathcal{A}_{2}$ may \emph{press} a button, add sugar and wait for coffee.
In $\mathcal{A}_{1}\parallel \mathcal{A}_{2}$, shared actions are performed synchronously
only if being enabled in both $\mathcal{A}_{1}$ and $\mathcal{A}_{2}$, thus resulting in a $\tau$-step.
For instance, the synchronized switch from \emph{idle} to \emph{add sugar} is 
labeled with $\tau$ and clocks $x$, $y$ (from $\mathcal{A}_{1}$) and $z$ (from $\mathcal{A}_{2}$) being reset.
Similarly, the \emph{sugar} loop also becomes a $\tau$-step, 
while clock resets are unified and guards are conjugated.
In contrast, switch \emph{proceed} does not become internal 
as this output is not observed by $\mathcal{A}_{2}$ (but instead
transmitted to some administration component).
Location \emph{preparing coffee} has two switches labeled with $\tau$ as 
both \emph{coffee} switches of $\mathcal{A}_{1}$ are synchronized with the 
\emph{coffee} switch of $\mathcal{A}_{2}$.
Location \emph{off} has a $\tau$-step to \emph{add sugar'} as the 
switch of $\mathcal{A}_{1}$ from \emph{off} to \emph{idle} may also be synchronized with 
the switch of $\mathcal{A}_{2}$ from \emph{idle} to \emph{add sugar}.
Here, \emph{add sugar'} does not have any outgoing transitions as \emph{add sugar} 
($\mathcal{A}_{2}$) has no actions shared with \emph{idle} ($\mathcal{A}_{1}$).
The \emph{sugar} loop and the switch from \emph{preparing coffee} 
to \emph{done} guarded by $y>15\land z\leq10$ are semantically incompatible as their guards are 
unsatisfiable in all runs.
\end{example}
%
\section{Timed Input/Output Conformance}\label{sec:tioco}
TIOLTS have been considered as a formal basis
for conformance testing theories
of time-critical input/output behaviors~\cite{Schmaltz2008}.
Timed conformance relations are usually defined
in the flavor of \textbf{ioco} testing, as initially proposed 
on input/output labeled transition systems (IOLTS)
for untimed behaviors~\cite{Tretmans1996}.

Intuitively, IOLTS $\textit{im}$ representing an implementation under test \emph{input/output-conforms} to
IOLTS $\textit{sp}$ representing a specification, denoted $\textit{im}~\textbf{ioco}~\textit{sp}$, 
if for all input behaviors specified in $\textit{sp}$, the observable
output behaviors of $\textit{im}$ for those input behaviors are permitted by $\textit{sp}$.
Input behaviors may be only partially
specified (\ie{} only for relevant/intended environmental input sequences, the expected
output behaviors are explicitly captured in $\textit{sp}$), whereas 
implementation $\textit{im}$ is supposed
to be \emph{input-enabled} (\ie{} to never block any input action).
Timed adaptations of \textbf{ioco}, so-called \textbf{tioco}, consider both $\textit{im}$ 
and $\textit{sp}$ to be represented as TIOLTS as
checking timed input/output conformance directly on TIOA is unfeasible
due to non-observability of clock resets in timed runs.
For instance, in the example in Fig.~\ref{fig:vending-machine-ta}, 
it is unknown if it is allowed to wait for 20 time units in \emph{idle} if we reach 
this location from \emph{done} as resets of $x$ and $y$ are not observable.
Similar to the untimed case, TIOLTS $\textit{im}$ is supposed to be input-enabled 
(\ie{} $\textit{im}$ must always---\emph{at any time}---be 
able to instantaneously accept all possible inputs).
In addition, for $\textit{im}$ to specify realistic behaviors, we further impose 
the \emph{independent-progress} property: 
In each state, $\textit{im}$ is able to either wait for an infinite amount of time 
or to eventually perform an output action thus preventing \emph{forced inputs}~\cite{David2010,Schmaltz2008}.

\begin{definition}\label{def:inp-en-ind-prog}
Let $(S,s_0,\Sigma_I,\Sigma_O,\twoheadrightarrow)$ be a TIOLTS.
\begin{itemize}
  \item (Input-Enabledness) 
  State $s\in S$ is \emph{weak input-enabled} iff $\forall i\in\Sigma_I:s\xTwoheadrightarrow{i}$.
  \item (Independent Progress)
  State $s\in S$ of a TIOLTS enables 
\emph{weak independent progress} iff 
$\forall d\in\Delta:s\xTwoheadrightarrow{d}$ or 
$\exists d\in\Delta,\exists o\in\Sigma_O:s\xTwoheadrightarrow{d}\xTwoheadrightarrow{o}$.
\end{itemize}
\end{definition}
A TIOLTS is \emph{(weak) input-enabled} iff all states are (weak) input-enabled and it enables 
\emph{(weak) independent progress} if all states do (for the strong versions of both properties, 
we replace $\Twoheadrightarrow$ by $\twoheadrightarrow$).
Similarly to \textbf{ioco}, we assume weak input-enabledness and 
independent progress for all implementations under test, whereas 
specifications may be underspecified.
This is required for practical testing where an implementation should always at least accept (and then potentially ignore) every input.
Conversely, the environment (\ie{} a tester) should not be enforced by the implementation to provide a particular input 
in order to guarantee any progress.
For instance, consider Fig.~\ref{fig:vending-machine-ta}: 
location \emph{off} is not input-enabled as there is no switch for input \emph{sugar}.
However, if there would be such a switch, then also location \emph{idle} would be weak 
input-enabled as output \emph{off} may be reached by a $\tau$-step.
In contrast, all locations in Fig.~\ref{fig:vending-machine-ta} 
enable (weak) independent progress.

We now revisit two major definitions of \textbf{tioco} from recent literature.
We first consider the (notationally slightly adapted) 
definition of Krichen and Tripakis~\cite{Krichen2004} which we will
refer to as $\mathtiocodelay$. 
It is based on the assumption that, in addition to
timed traces consisting of sequences of timed steps $(d,o)$ including 
output actions $o\in \Sigma_O$, also
all possible delays $d\in \Delta$ permitted 
to elapse in states $s\in S$ are observable in isolation.

\begin{definition}[$\mathtiocodelay$]\label{def:tioco-delay}
    Let \emph{im}, \emph{sp} be a TIOLTS over
    $\Sigma=\Sigma_I\cup\Sigma_O$, $s\in S$, $S'\subseteq S$, 
    and $\xi\in(\Delta\times\Sigma)^*$.
    \begin{itemize}
        \item $s\after\xi:=\{s'\mid s\xTwoheadrightarrow{\xi}s'\}$,
        \item $\outdelay(s):=\{o\mid o\in\Sigma_O, s\xtwoheadrightarrow{o}\}\cup\{d\mid d\in\Delta, s\xTwoheadrightarrow{d}\}$,
        \item $\outdelay(S'):=\bigcup_{s\in S'}\outdelay(s)$,
        \item $\ttraces(s):=\{\xi\mid s\xTwoheadrightarrow{\xi}\}$, and
      \item $\textit{im}~\mathtiocodelay~\textit{sp}:\Leftrightarrow \forall\xi\in\ttraces(\textit{sp}):\outdelay(\textit{im}\after\xi)\subseteq\outdelay(\textit{sp}\after\xi)$
      \end{itemize}
\end{definition}
We may use the name of the whole TIOLTS and the name of its initial state interchangeably 
as frequently done in \textbf{ioco}-based theories (\eg{} by $\textit{im}\,\after\,\xi$ we refer to the set of 
states being reachable by 
$\xi$ from the initial state of \textit{im}).
The second version of \textbf{tioco}, which we will denote as $\mathtiocodelta$,
does not rely on observability of arbitrary delays, but
instead incorporates a notion of \emph{timed quiescence}~\cite{Schmaltz2008}. 
Quiescence constitutes another fundamental concept of (untimed) \textbf{ioco}: 
IOLTS state $s$ is \emph{quiescent}, denoted $\delta(s)$, if no output or internal action is enabled in $s$ thus
requiring an input to proceed a (suspended) run reaching $s$.
By making quiescence observable by a special output $\delta$, 
\textbf{ioco} rejects trivial implementations $\textit{im}$ never showing any outputs 
as this must be explicitly permitted by the specification.
In the timed case, state $s$ of a TIOLTS may be considered
quiescent if no output action is \emph{ever} (or, at least not until some fixed maximum delay $M$~\cite{BrandanBriones2004}) enabled in $s$.
To this end, the notion of \emph{timed suspension traces} (\textit{tstraces})
extends \emph{traces} of TIOLTS by timed observable quiescence.
The most common definition of $\mathtiocodelta$ may be given as follows.

\begin{definition}[$\mathtiocodelta$]\label{def:tioco-auxiliary}
    Let \emph{im}, \emph{sp} be a TIOLTS over
    $\Sigma=\Sigma_I\cup\Sigma_O$, $s,s'\in S$, $S'\subseteq S$ 
    and $\xi\in(\Delta\times(\Sigma\cup\{\delta\}))^*$.
    \begin{itemize}
        \item $s$ is \emph{quiescent}, denoted by $\delta(s)$, iff $\forall\mu\in\Sigma_O,\forall d\in\Delta:s\not\xtwoheadrightarrow{(d,\mu)}$,
        \item $s\after\xi:=\{s'\mid s\xTwoheadrightarrow{\xi}s'\}$,
        \item $\out(s):=\{(d,o)\mid o\in\Sigma_O,d\in\Delta,s\xTwoheadrightarrow{(d,o)}\}\cup\{\delta\mid\delta(s)\}$,
        \item $\out(S'):=\bigcup_{s\in S'}\out(s)$,
        \item $\tstraces(s):=\{\xi\mid s\xTwoheadrightarrow{\xi}\}$, where $s'\xtwoheadrightarrow{\delta}s'$ iff $\delta(s')$, and
        \item $\textit{im}~\mathtiocodelta~\textit{sp}:\Leftrightarrow \forall\xi\in\tstraces(\textit{sp}):\out(\textit{im}\after\xi)\subseteq\out(\textit{sp}\after\xi)$
        \end{itemize}
\end{definition}

\begin{example}
Figure~\ref{fig:tioco-examples} provides a collection of small examples illustrating $\mathtiocodelta$.
In Fig.~\ref{fig:tioco-example-correct}, it holds that $\llbracket A_0\rrbracket_S\mathtiocodelta\llbracket A_1\rrbracket_S$ as the required inclusion relation holds for all possible $\out$ sets, for instance, $\out(\llbracket A_0\rrbracket_S\after\epsilon)=\{(1,o),(2,o)\}\subseteq\out(\llbracket A_1\rrbracket_S\after\epsilon)=\{(1,o),(2,o),(3,o)\}$.
Note, that this is also true for output behaviors enabled after 3 time units as $\llbracket A_0\rrbracket_S$ does not permit to wait for 3 time units, such that the respective $\out$ set is empty.
Hence, \textbf{tioco} permits implementations to show less output behavior than the specification allows.
\begin{figure}[tp]
    \subfloat[$\llbracket A_0\rrbracket_S\mathtiocodelta\llbracket A_1\rrbracket_S$]{\hspace{.1em}\label{fig:tioco-example-correct}

\scalebox{\scalefactor}{
\begin{tikzpicture}[node distance=.5]

\node[labeledstate, initial, initial where=above] (l00) {$\angles{\ell_0,x=0}$};
\node[labeledstate, below=of l00] (l01) {$\angles{\ell_0,x=1}$};
\node[labeledstate, below=of l01] (l02) {$\angles{\ell_0,x=2}$};

\draw[live] (l00) to node[auto] {$1$} (l01);
\draw[live] (l01) to node[auto] {$1$} (l02);
\draw[live, bend left=60] (l01) to node[pos=.4, right] {!o} (l00);
\draw[live, bend left=60] (l02) to node[pos=.2, right] {!o} (l00);

\node[labeledstate, initial, initial where=above, right=of l00] (l10) {$\angles{\ell_1,x=0}$};
\node[labeledstate, below=of l10] (l11) {$\angles{\ell_1,x=1}$};
\node[labeledstate, below=of l11] (l12) {$\angles{\ell_1,x=2}$};
\node[labeledstate, below=of l12] (l13) {$\angles{\ell_1,x=3}$};

\draw[live] (l10) to node[auto] {$1$} (l11);
\draw[live] (l11) to node[auto] {$1$} (l12);
\draw[live] (l12) to node[auto] {$1$} (l13);
\draw[live, bend left=60] (l11) to node[pos=.4, right] {!o} (l10);
\draw[live, bend left=60] (l12) to node[pos=.2, right] {!o} (l10);
\draw[live, bend left=60] (l13) to node[pos=.2, right] {!o} (l10);

\end{tikzpicture}
}\hspace{.1em}}
    \hfill
    \subfloat[$\llbracket A_2\rrbracket_S\not\mathtiocodelta\llbracket A_3\rrbracket_S$]{\label{fig:tioco-example-incorrect}

\scalebox{\scalefactor}{
\begin{tikzpicture}[node distance=.5]

\node[labeledstate, initial, initial where=above] (l40) {$\angles{\ell_2,x=0}$};
\node[labeledstate, left=of l40] (l40empty) {};
\node[labeledstate, below=of l40] (l41) {$\angles{\ell_2,x=1}$};
\node[labeledstate, below=of l41] (l42) {$\angles{\ell_2,x=2}$};
\node[labeledstate, below=of l42] (l43) {$\vdots$};

\draw[live] (l40) to node[auto] {$1$} (l41);
\draw[live] (l41) to node[auto] {$1$} (l42);
\draw[live] (l42) to node[auto] {$1$} (l43);

\node[labeledstate, initial, initial where=above, right=of l40] (l50) {$\angles{\ell_3,x=0}$};
\node[labeledstate, right=of l50] (l50empty) {};
\node[labeledstate, below=of l50] (l51) {$\angles{\ell_3,x=1}$};
\node[labeledstate, below=of l51] (l52) {$\angles{\ell_3,x=2}$};
\node[labeledstate, below=of l52] (l53) {$\vdots$};

\draw[live] (l50) to node[auto] {$1$} (l51);
\draw[live] (l51) to node[auto] {$1$} (l52);
\draw[live] (l52) to node[auto] {$1$} (l53);
\draw[live, bend left=60] (l51) to node[pos=.4, right] {!o} (l50);
\draw[live, bend left=60] (l52) to node[pos=.2, right] {!o} (l50);
\draw[live, bend left=60] (l53) to node[pos=.2, right] {!o} (l50);

\end{tikzpicture}
}}
    \hfill
    \subfloat[$\llbracket A_4\rrbracket_S\mathtiocodelta\llbracket A_5\rrbracket_S$]{\label{fig:tioco-example-nondet}

\scalebox{\scalefactor}{
\begin{tikzpicture}[node distance=.5]

\node[labeledstate, initial, initial where=above] (l20) {$\angles{\ell_4,x=0}$};
\node[labeledstate, below left=of l20] (l2p0) {$\angles{\ell_4',x=0}$};
\node[labeledstate, below=of l2p0] (l2p1) {$\angles{\ell_4',x=1}$};
\node[labeledstate, below=of l2p1] (l2p2) {$\angles{\ell_4',x=2}$};

\node[labeledstate, below right=of l20] (l2pp0) {$\angles{\ell_4'',x=0}$};
\node[labeledstate, below=of l2pp0] (l2pp1) {$\angles{\ell_4'',x=1}$};
\node[labeledstate, below=of l2pp1] (l2pp2) {$\angles{\ell_4'',x=2}$};

\draw[live, bend right=30] (l20) to node[auto, swap] {!o} (l2p0);
\draw[live] (l2p0) to node[auto] {$1$} (l2p1);
\draw[live] (l2p1) to node[auto] {$1$} (l2p2);
\draw[live] (l2p0) to node[pos=.1, right] {!o} (l20);
\draw[live, bend right=20] (l2p1) to node[pos=.4, left] {!o} (l20);
\draw[live, bend right=20] (l2p2) to node[pos=.4, left] {!o} (l20);

\draw[live, bend left=30] (l20) to node[auto] {!o} (l2pp0);
\draw[live] (l2pp0) to node[auto] {$1$} (l2pp1);
\draw[live] (l2pp1) to node[auto] {$1$} (l2pp2);
\draw[live] (l2pp0) to node[pos=.1, left] {!o} (l20);
\draw[live, bend left=20] (l2pp1) to node[pos=.4, right] {!o} (l20);
\draw[live, bend left=20] (l2pp2) to node[pos=.4, right] {!o'} (l20);

\node[labeledstate, right=of l2pp0] (l31) {$\angles{\ell_5',x=0}$};
\node[labeledstate, initial, initial where=above, above=of l31] (l30) {$\angles{\ell_5,x=0}$};
\node[labeledstate, below=of l31] (l32) {$\angles{\ell_5',x=1}$};
\node[labeledstate, below=of l32] (l33) {$\angles{\ell_5',x=2}$};

\draw[live] (l30) to node[auto] {!o} (l31);
\draw[live] (l31) to node[auto] {$1$} (l32);
\draw[live] (l32) to node[auto] {$1$} (l33);
\draw[live, bend left=60] (l31) to node[pos=.4, right] {!o} (l30);
\draw[live, bend left=60] (l32) to node[pos=.2, right] {!o} (l30);
\draw[live, bend left=60] (l33) to node[pos=.2, right] {!o} (l30);
\draw[live, bend right=70] (l33) to node[pos=.2, left] {!o'} (l30);

\end{tikzpicture}
}}
  \caption{Examples for $\mathtiocodelta$ on TIOLS}\label{fig:tioco-examples}
\end{figure}
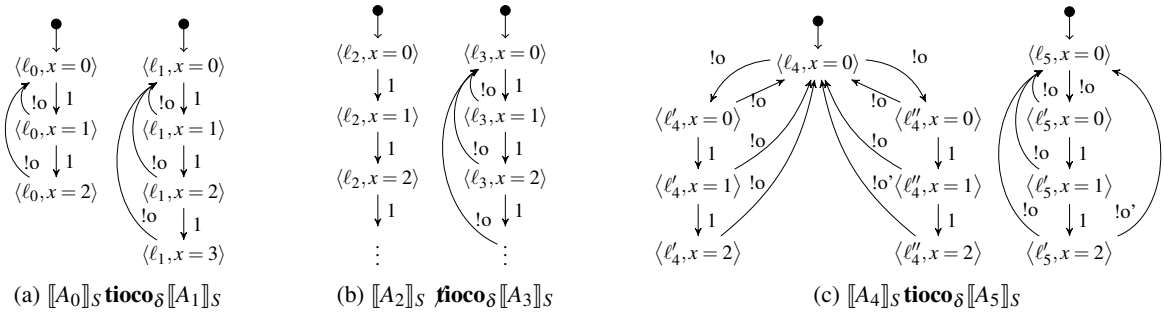
Figure~\ref{fig:tioco-example-incorrect} depicts a further example where $\llbracket A_2\rrbracket_S\mathtiocodelta\llbracket A_3\rrbracket_S$ does not hold as $\out(\llbracket A_2\rrbracket_S\after\epsilon)=\{\delta\}\nsubseteq\out(\llbracket A_3\rrbracket_S\after\epsilon)=\{(1,o),(2,o),\ldots\}$ (\ie{} implementation $\llbracket A_2\rrbracket_S$ is quiescent but specification $\llbracket A_3\rrbracket_S$ is not).
The TIOLTS in Fig.~\ref{fig:tioco-example-nondet} illustrates how non-determinism is handled by $\mathtiocodelta$.
For specification ($\llbracket A_5\rrbracket_S$), it holds that $\out(\llbracket A_5\rrbracket_S\after(0,o))=\{(0,o),(1,o),\allowbreak(2,o),\allowbreak(2,o')\}$, and the same holds for $\out(\llbracket A_4\rrbracket_S\after(0,o))$ and, particularly, $(2,o)$ \emph{and} $(2,o')$.
This is due to the fact that in $\mathtiocodelta$ only outputs of those states are considered 
being reachable via some trace of the specification, but not necessarily of \emph{any} state of the respective TIOLTS.
Therefore, it holds that $\llbracket A_4\rrbracket_S\mathtiocodelta\llbracket A_5\rrbracket_S$.
\end{example}

Next, we apply $\mathtiocodelta$ to our running example to illustrate the differences to $\mathtiocodelay$.

\begin{example}
Consider TIOA $\mathcal{A}_1'$ depicted in Fig.~\ref{fig:vending-machine-tioco} to be a candidate implementation of TIOA $\mathcal{A}_1$ in Fig.~\ref{fig:vending-machine-ta}.
\begin{figure}[tp]
  \centering
  \begin{adjustbox}{max width=\textwidth}

\scalebox{\scalefactor}{
\begin{tikzpicture}[node distance=1.2]

\node[labeledroundedstateTA, initial, initial where=below] (idle) {idle};
\node[labeledroundedstateTA, left=of idle] (off) {off};
\node[labeledroundedstateTA, right=of idle] (as) {add\\sugar\\$y\leq 15$};
\node[labeledroundedstateTA, above=of as] (pc) {preparing\\coffee\\$y\leq20$};
\node[labeledroundedstateTA, left=of pc] (done) {done\\$y\leq20$};

\draw[live] (idle) to node[pos=.5, align=center] {$\tau$\\$[x=20]$} (off);
\draw[live] (idle) to node[pos=.5, align=center] {?press\\$[x\leq20]$\\$x:=0$\\$y:=0$} (as);
\draw[live,loop below] (as) to node[pos=.25, right, align=left] {?sugar\\$[x\geq10]$\\$x:=0$} (as);
\draw[live] (as) to node[pos=.5, right, align=left] {!proceed\\$y:=0$} (pc);
\draw[live, bend right=25] (pc) to node[pos=.5, align=center] {!coffee\\$[y>15]$} (done);
\draw[live, bend left=25] (pc) to node[pos=.5, align=center] {!coffee\\$[y\leq15]$} (done);
\draw[live] (done) to node[auto, left, align=right] {$\tau$\\$x:=0$\\$y:=0$} (idle);
\draw[live, bend right=45] (off) to node[pos=.53, align=center] {\\\\?press\\$x:=0$\\$y:=0$} (idle);

\end{tikzpicture}
}
  \end{adjustbox}
  \caption{Example for a Candidate Implementation $\mathcal{A}_1'$ of $\mathcal{A}_1$ (\cf{}~Fig.~\ref{fig:vending-machine-ta})}\label{fig:vending-machine-tioco}
\end{figure}
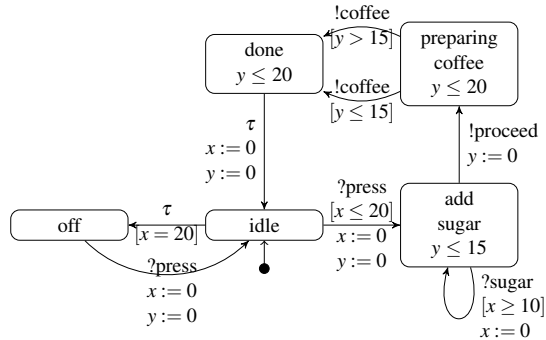
First, the guard $y\leq20$ of the switch labeled with \emph{proceed} is not contained 
in $\mathcal{A}_1'$, and instead, location \emph{add sugar} has an invariant $y\leq15$.
Considering only this difference, we have $\llbracket\mathcal{A}_1'\rrbracket_S\mathtiocodelta\llbracket\mathcal{A}_1\rrbracket_S$ as well as $\llbracket\mathcal{A}_1'\rrbracket_S\mathtiocodelay\llbracket\mathcal{A}_1\rrbracket_S$ as we forbid output \emph{proceed} for $15<y\leq20$ and waiting in \emph{add sugar} for an arbitrary amount of time.
In contrast, omitting the switch labeled \emph{!proceed} in $\mathcal{A}_1'$ would lead to a violation of $\mathtiocodelta$ as location \emph{add sugar} would become quiescent (whereas $\mathtiocodelay$ still holds as it does not check for quiescence).
Second, the invariant $x\leq20$ of location \emph{idle} in $\mathcal{A}_1$ is not contained in $\mathcal{A}_1'$ but, instead, becomes a guard to the switch labeled with \emph{?press}.
As a result, $\llbracket\mathcal{A}_1'\rrbracket_S\mathtiocodelta\llbracket\mathcal{A}_1\rrbracket_S$ still holds as delays in timed runs are only observable by $\mathtiocodelta$ if paired with a subsequent output action.
In contrast, $\llbracket\mathcal{A}_1'\rrbracket_S\mathtiocodelay\llbracket\mathcal{A}_1\rrbracket_S$ does not hold as in $\mathtiocodelay$ delays of any possible duration are observable, even if no subsequent outputs will ever occur.
\end{example}

\paragraph{Weaknesses of Existing Definitions of Timed Input/Output Conformance.}
As a result, $\mathtiocodelay$ and $\mathtiocodelta$ are incomparable.
In addition, observability capabilities required for effectively checking $\mathtiocodelay$
are unrealistic and therefore only of theoretical interest, but infeasible in practice. 
In contrast, $\mathtiocodelta$ is more realistic but fails to guarantee liveness requirements as 
the notion of quiescence does not properly
reflect the differences between allowed and enforced outputs in TIOA specifications.
To further illustrate this problem, consider the five
TIOA, $\mathcal{A}_1$ to $\mathcal{A}_5$, and their TIOLTS in Fig.~\ref{fig:tioco-flaws}.
\begin{figure}[tp]
  \newcommand{\nodeDistanceTA}{1}
  \newcommand{\nodeDistanceTLTS}{.3}
  \newcommand{\figureScaling}{\scalefactor}
  \begin{adjustbox}{max width=\textwidth}
    \begin{tabular}{lll}
      \toprule

\specialcell[l]{$\mathcal{A}_1$\\enforced $\delta$\\forbidden $(d,o)$}

&

\specialcell{
	\begin{tikzpicture}[node distance=\nodeDistanceTA]

		\node[labeledroundedstateTA, initial, initial where=left] (l1) {$\ell_1$};

	\end{tikzpicture}
} 

&

\specialcell[l]{
		\begin{tikzpicture}[node distance=\nodeDistanceTLTS]

			\node[labeledstate] (l10) {$\angles{\ell_1,x=0}$};
			\node[labeledstate, right=of l10] (l1d) {$\angles{\ell_1,x=d}$};

			\draw[live] (l10) to node[auto] {$d$} (l1d);

		\end{tikzpicture}
	\\
		\begin{tikzpicture}[node distance=\nodeDistanceTLTS]

			\node[labeledstate] (l1u) {$\angles{\ell_1,u}$};
			\node[labeledstate, right=of l1u] (l1uu) {$\angles{\ell_1,u}$};

			\draw[live] (l1u) to node[auto] {$\delta$} (l1uu);

		\end{tikzpicture}
}

\\\midrule

\specialcell[l]{$\mathcal{A}_2$\\safe $\delta$\\safe $(d,o)$}

&

\specialcell{
	\begin{tikzpicture}[node distance=\nodeDistanceTA]

		\node[labeledroundedstateTA, initial, initial where=left] (l2) {$\ell_2$};
		\node[labeledroundedstateTA, right=of l2] (l2p) {$\ell_2'$};

		\draw[live] (l2) to node[auto] {!o} (l2p);

	\end{tikzpicture}
}

&

\specialcell[l]{
		\begin{tikzpicture}[node distance=\nodeDistanceTLTS]

			\node[labeledstate] (l20) {$\angles{\ell_2,x=0}$};
			\node[labeledstate, right=of l20] (l2d) {$\angles{\ell_2,x=d}$};

			\draw[live] (l20) to node[auto] {$d$} (l2d);

		\end{tikzpicture}
	\\
		\begin{tikzpicture}[node distance=\nodeDistanceTLTS]

			\node[labeledstate] (l20) {$\angles{\ell_2,x=0}$};
			\node[labeledstate, right=of l20] (l2d) {$\angles{\ell_2,x=d}$};
			\node[labeledstate, right=of l2d] (l2p) {$\angles{\ell_2',x=d}$};

			\draw[live] (l20) to node[auto] {$d$} (l2d);
			\draw[live] (l2d) to node[auto] {!o} (l2p);

		\end{tikzpicture}
}

\\\midrule

\specialcell[l]{$\mathcal{A}_3$\\safe $\delta$\\safe $(d<k,o)$}

&

\specialcell{
	\begin{tikzpicture}[node distance=\nodeDistanceTA]

		\node[labeledroundedstateTA, initial, initial where=left] (l3) {$\ell_3$};
		\node[labeledroundedstateTA, right=of l3] (l3p) {$\ell_3'$};

		\draw[live] (l3) to node[pos=.5,align=center] {!o\\$x<k$} (l3p);

	\end{tikzpicture}
}

&

\specialcell[l]{
		\begin{tikzpicture}[node distance=\nodeDistanceTLTS]

			\node[labeledstate] (l30) {$\angles{\ell_3,x=0}$};
			\node[labeledstate, right=of l30] (l3d) {$\angles{\ell_3,x=d}$};

			\draw[live] (l30) to node[auto] {$d$} (l3d);

		\end{tikzpicture}
	\\
		\begin{tikzpicture}[node distance=\nodeDistanceTLTS]

			\node[labeledstate] (l30) {$\angles{\ell_3,x=0}$};
			\node[labeledstate, right=of l30, xshift=15] (l3d) {$\angles{\ell_3,x=d}$};
			\node[labeledstate, right=of l3d] (l3p) {$\angles{\ell_3',x=d}$};

			\draw[live] (l30) to node[auto] {$d<k$} (l3d);
			\draw[live] (l3d) to node[auto] {!o} (l3p);

		\end{tikzpicture}
}

\\\midrule

\specialcell[l]{$\mathcal{A}_4$\\forbidden $\delta$\\enforced $(d<k,o)$}

&

\specialcell{
	\begin{tikzpicture}[node distance=\nodeDistanceTA]

		\node[labeledroundedstateTA, initial, initial where=left] (l4) {$\ell_4$\\$x<k$};
		\node[labeledroundedstateTA, right=of l4] (l4p) {$\ell_4'$};

		\draw[live] (l4) to node[auto] {!o} (l4p);

	\end{tikzpicture}
}

&

\specialcell{
	\begin{tikzpicture}[node distance=\nodeDistanceTLTS]

		\node[labeledstate] (l40) {$\angles{\ell_4,x=0}$};
		\node[labeledstate, right=of l40, xshift=15] (l4d) {$\angles{\ell_4,x=d}$};
		\node[labeledstate, right=of l4d] (l4p) {$\angles{\ell_4',x=d}$};

		\draw[live] (l40) to node[auto] {$d<k$} (l4d);
		\draw[live] (l4d) to node[auto] {!o} (l4p);

	\end{tikzpicture}
}

\\\midrule

\specialcell[l]{$\mathcal{A}_5$\\forbidden $\delta$\\enforced $(d<k,o)$}

&

\specialcell{
	\begin{tikzpicture}[node distance=\nodeDistanceTA]

		\node[labeledroundedstateTA, initial, initial where=left] (l5) {$\ell_5$\\$x<k$};
		\node[labeledroundedstateTA, right=of l5] (l5p) {$\ell_5'$};

		\draw[live] (l5) to node[pos=.5, align=center] {!o\\$x<k$} (l5p);

	\end{tikzpicture}
}

&

\specialcell{
	\begin{tikzpicture}[node distance=\nodeDistanceTLTS]

		\node[labeledstate] (l50) {$\angles{\ell_5,x=0}$};
		\node[labeledstate, right=of l50, xshift=15] (l5d) {$\angles{\ell_5,x=d}$};
		\node[labeledstate, right=of l5d] (l5p) {$\angles{\ell_5',x=d}$};

		\draw[live] (l50) to node[auto] {$d<k$} (l5d);
		\draw[live] (l5d) to node[auto] {!o} (l5p);

	\end{tikzpicture}
}

\\
      \bottomrule
    \end{tabular}
    \setbox9=\hbox{\begin{tabular}{lll}
      \toprule

\specialcell[l]{$\mathcal{A}_1$\\enforced $\delta$\\forbidden $(d,o)$}

&

\specialcell{
	\begin{tikzpicture}[node distance=\nodeDistanceTA]

		\node[labeledroundedstateTA, initial, initial where=left] (l1) {$\ell_1$};

	\end{tikzpicture}
} 

&

\specialcell[l]{
		\begin{tikzpicture}[node distance=\nodeDistanceTLTS]

			\node[labeledstate] (l10) {$\angles{\ell_1,x=0}$};
			\node[labeledstate, right=of l10] (l1d) {$\angles{\ell_1,x=d}$};

			\draw[live] (l10) to node[auto] {$d$} (l1d);

		\end{tikzpicture}
	\\
		\begin{tikzpicture}[node distance=\nodeDistanceTLTS]

			\node[labeledstate] (l1u) {$\angles{\ell_1,u}$};
			\node[labeledstate, right=of l1u] (l1uu) {$\angles{\ell_1,u}$};

			\draw[live] (l1u) to node[auto] {$\delta$} (l1uu);

		\end{tikzpicture}
}

\\\midrule

\specialcell[l]{$\mathcal{A}_2$\\safe $\delta$\\safe $(d,o)$}

&

\specialcell{
	\begin{tikzpicture}[node distance=\nodeDistanceTA]

		\node[labeledroundedstateTA, initial, initial where=left] (l2) {$\ell_2$};
		\node[labeledroundedstateTA, right=of l2] (l2p) {$\ell_2'$};

		\draw[live] (l2) to node[auto] {!o} (l2p);

	\end{tikzpicture}
}

&

\specialcell[l]{
		\begin{tikzpicture}[node distance=\nodeDistanceTLTS]

			\node[labeledstate] (l20) {$\angles{\ell_2,x=0}$};
			\node[labeledstate, right=of l20] (l2d) {$\angles{\ell_2,x=d}$};

			\draw[live] (l20) to node[auto] {$d$} (l2d);

		\end{tikzpicture}
	\\
		\begin{tikzpicture}[node distance=\nodeDistanceTLTS]

			\node[labeledstate] (l20) {$\angles{\ell_2,x=0}$};
			\node[labeledstate, right=of l20] (l2d) {$\angles{\ell_2,x=d}$};
			\node[labeledstate, right=of l2d] (l2p) {$\angles{\ell_2',x=d}$};

			\draw[live] (l20) to node[auto] {$d$} (l2d);
			\draw[live] (l2d) to node[auto] {!o} (l2p);

		\end{tikzpicture}
}

\\\midrule

\specialcell[l]{$\mathcal{A}_3$\\safe $\delta$\\safe $(d<k,o)$}

&

\specialcell{
	\begin{tikzpicture}[node distance=\nodeDistanceTA]

		\node[labeledroundedstateTA, initial, initial where=left] (l3) {$\ell_3$};
		\node[labeledroundedstateTA, right=of l3] (l3p) {$\ell_3'$};

		\draw[live] (l3) to node[pos=.5,align=center] {!o\\$x<k$} (l3p);

	\end{tikzpicture}
}

&

\specialcell[l]{
		\begin{tikzpicture}[node distance=\nodeDistanceTLTS]

			\node[labeledstate] (l30) {$\angles{\ell_3,x=0}$};
			\node[labeledstate, right=of l30] (l3d) {$\angles{\ell_3,x=d}$};

			\draw[live] (l30) to node[auto] {$d$} (l3d);

		\end{tikzpicture}
	\\
		\begin{tikzpicture}[node distance=\nodeDistanceTLTS]

			\node[labeledstate] (l30) {$\angles{\ell_3,x=0}$};
			\node[labeledstate, right=of l30, xshift=15] (l3d) {$\angles{\ell_3,x=d}$};
			\node[labeledstate, right=of l3d] (l3p) {$\angles{\ell_3',x=d}$};

			\draw[live] (l30) to node[auto] {$d<k$} (l3d);
			\draw[live] (l3d) to node[auto] {!o} (l3p);

		\end{tikzpicture}
}

\\\midrule

\specialcell[l]{$\mathcal{A}_4$\\forbidden $\delta$\\enforced $(d<k,o)$}

&

\specialcell{
	\begin{tikzpicture}[node distance=\nodeDistanceTA]

		\node[labeledroundedstateTA, initial, initial where=left] (l4) {$\ell_4$\\$x<k$};
		\node[labeledroundedstateTA, right=of l4] (l4p) {$\ell_4'$};

		\draw[live] (l4) to node[auto] {!o} (l4p);

	\end{tikzpicture}
}

&

\specialcell{
	\begin{tikzpicture}[node distance=\nodeDistanceTLTS]

		\node[labeledstate] (l40) {$\angles{\ell_4,x=0}$};
		\node[labeledstate, right=of l40, xshift=15] (l4d) {$\angles{\ell_4,x=d}$};
		\node[labeledstate, right=of l4d] (l4p) {$\angles{\ell_4',x=d}$};

		\draw[live] (l40) to node[auto] {$d<k$} (l4d);
		\draw[live] (l4d) to node[auto] {!o} (l4p);

	\end{tikzpicture}
}

\\\midrule

\specialcell[l]{$\mathcal{A}_5$\\forbidden $\delta$\\enforced $(d<k,o)$}

&

\specialcell{
	\begin{tikzpicture}[node distance=\nodeDistanceTA]

		\node[labeledroundedstateTA, initial, initial where=left] (l5) {$\ell_5$\\$x<k$};
		\node[labeledroundedstateTA, right=of l5] (l5p) {$\ell_5'$};

		\draw[live] (l5) to node[pos=.5, align=center] {!o\\$x<k$} (l5p);

	\end{tikzpicture}
}

&

\specialcell{
	\begin{tikzpicture}[node distance=\nodeDistanceTLTS]

		\node[labeledstate] (l50) {$\angles{\ell_5,x=0}$};
		\node[labeledstate, right=of l50, xshift=15] (l5d) {$\angles{\ell_5,x=d}$};
		\node[labeledstate, right=of l5d] (l5p) {$\angles{\ell_5',x=d}$};

		\draw[live] (l50) to node[auto] {$d<k$} (l5d);
		\draw[live] (l5d) to node[auto] {!o} (l5p);

	\end{tikzpicture}
}

\\
      \bottomrule
    \end{tabular}}
    \hspace*{.5em}
    \raisebox{\dimexpr-\ht9+\height}{

\begin{tabular}{c|cccccccl}
\toprule

\backslashbox{$\textit{im}$}{$\textit{sp}$} & $\llbracket\mathcal{A}_1\rrbracket_S$ & $\llbracket\mathcal{A}_2\rrbracket_S$ & $\llbracket\mathcal{A}_3\rrbracket_S$ & $\llbracket\mathcal{A}_4\rrbracket_S$ & $\llbracket\mathcal{A}_5\rrbracket_S$ \\\midrule
$\llbracket\mathcal{A}_1\rrbracket_S$ & \cmark & \xmark & \xmark & \xmark & \xmark \\
$\llbracket\mathcal{A}_2\rrbracket_S$ & \xmark & \cmark & \xmark & \xmark & \xmark \\
$\llbracket\mathcal{A}_3\rrbracket_S$ & \xmark & \cmark & \cmark & \cellcolor{gray!75}\cmark & \cellcolor{gray!75}\cmark \\
$\llbracket\mathcal{A}_4\rrbracket_S$ & \xmark & \cmark & \cmark & \cmark & \cmark \\
$\llbracket\mathcal{A}_5\rrbracket_S$ & \xmark & \cmark & \cmark & \cmark & \cmark \\

\bottomrule
\end{tabular}}
  \end{adjustbox}
  \caption{Examples of Allowed/Safe, Enforced and Forbidden Actions (TIOA depicted on the left, TIOLTS depicted on the right).
  Table on the right lists for all Pairwise Combinations of TIOA whether $\mathtiocodelta$ holds.}\label{fig:tioco-flaws}
\end{figure}
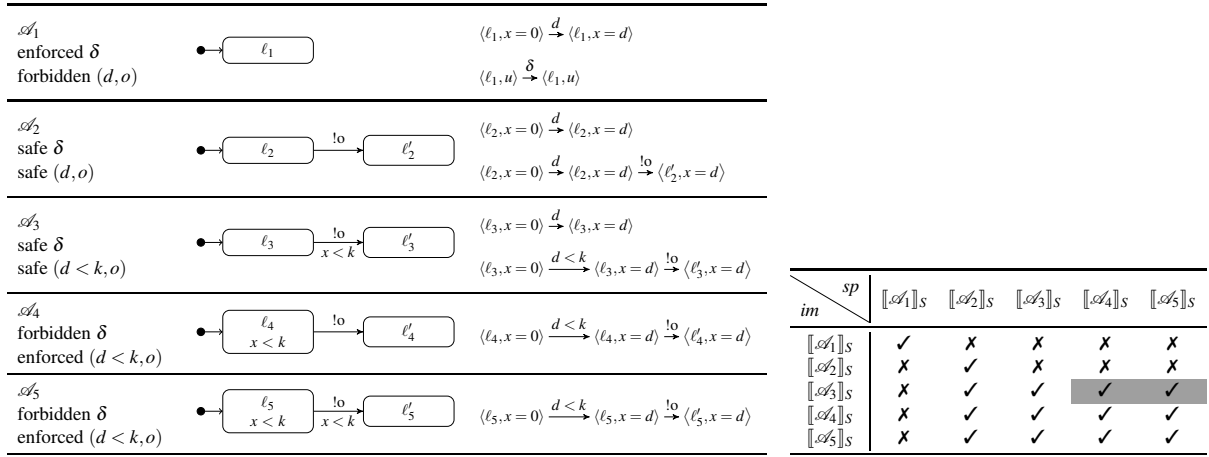
According to Def.~\ref{def:tioco-auxiliary}, 
location $\ell_{1}$ of $\mathcal{A}_1$ is quiescent, whereas none of the locations 
$\ell_{2}$ to $\ell_5$ of $\mathcal{A}_2$ to $\mathcal{A}_5$
are quiescent as output $o$ is eventually enabled.
The table in Fig.~\ref{fig:tioco-flaws} shows all possible comparisons of all five TIOA under $\mathtiocodelta$.
Here, the fact that 
$\llbracket\mathcal{A}_3\rrbracket_S\,\mathtiocodelta\,\llbracket\mathcal{A}_4\rrbracket_S$ and
$\llbracket\mathcal{A}_3\rrbracket_S\,\mathtiocodelta\,\llbracket\mathcal{A}_5\rrbracket_S$
hold is particularly undesirable (as highlighted in the table):
$\mathcal{A}_3$ \emph{may} either produce output $o$ within interval $0 \leq x < k$, or it \emph{may} behave quiescent, whereas $\mathcal{A}_4$ and $\mathcal{A}_5$ \emph{must} produce output $o$ within interval $0 \leq x < k$ and therefore \emph{must not} be quiescent.
In contrast, $\mathcal{A}_2$ and $\mathcal{A}_3$ are allowed 
to be quiescent, by residing for unlimited durations in $\ell_2$ and $\ell_3$.

We summarize the most important weaknesses of existing
versions of \textbf{tioco}.

\begin{itemize}
  \item \emph{(Live Timed Behaviors)}
  \textbf{tioco} either relies on a (unrealistically) strong
  notion of observability including arbitrary delays, or on a (unnecessarily) weak 
  notion of quiescence not distinguishing allowed from enforced outputs.

  \item \emph{(Compositionality)} 
  To the best of our knowledge, there only exists one work
  investigating compositionality properties of \textbf{tioco} so far
  which does not take any notion of quiescence into account~\cite{Bannour2013}.
  
  \item \emph{(Infinite TIOLTS)}
  \textbf{tioco} is defined on TIOLTS,
  an infinitely-bran\-ching state-transition graph being
intractable for realistic testing practices and tools.
However, a sound characterization of \textbf{tioco} directly on TIOA is also not feasible
as timed (suspension) traces are not directly derivable from TIOA.
\end{itemize}
We next propose an improved version of \textbf{tioco} 
to tackle these weaknesses.
%
\section{Improved Timed Input/Output Conformance}\label{sec:ltioco}

In this section, we tackle the weaknesses of existing versions of \textbf{tioco} as described in the previous section.

\subsection{Safe vs. Enforced Quiescence}
Existing definitions of $\textbf{tioco}$ 
either do not have any notion of quiescence at all~\cite{Krichen2004}, or
quiescence includes both (1) states that, if no input is provided, will delay 
forever with no output and (2) states that may eventually produce an output (\cf{}~Fig.~\ref{fig:tioco-flaws})~\cite{Schmaltz2008}.
We instead consider two different facets of quiescence: 
state $s$ is \emph{enforced quiescent} if each run \emph{must} wait in this state 
for an input for an arbitrary duration to proceed.
This coincides with quiescence of $\mathtiocodelta$.
In contrast, state $s$ is \emph{safe quiescent} if a run \emph{may} wait in this state
for an input for an arbitrary duration, but \emph{may} also proceed by eventually producing an output.
Consequently, state $s$ is \emph{not quiescent}, if a run \emph{must} eventually
proceed from this state by producing an output.
Hence, $s$ is \emph{live} if it is neither safe quiescent nor enforced quiescent.

\begin{definition}[Safe/Enforced Quiescence]\label{def:ltioco-quiescence}
    Let $(S,s_0,\Sigma_I,\Sigma_O,\twoheadrightarrow)$ be a TIOLTS.
    \begin{itemize}
      \item $s\in S$ is \emph{safe-quiescent}, denoted $\delta_{S}(s)$, 
      iff $\forall d\in\Delta:s\xtwoheadrightarrow{d}$.
      \item $s\in S$ is \emph{enforced-quiescent}, denoted $\delta_{E}(s)$, 
      iff $\forall\mu\in\Sigma_O,\forall d\in\Delta:s\not\xtwoheadrightarrow{(d,\mu)}$.
    \end{itemize}
\end{definition}
Intuitively, we may assume enforced-quiescent states to be also safe-quiescent.
However, as a counter-example, assume a TIOLTS with one state $\angles{\ell,x=0}$ (corresponding to a TIOA with 
one location $\ell$ and $I(\ell)=x\leq0$): here, no outputs 
are possible and no delays are allowed thus obstructing the intuition.

\begin{lemma}\label{lemma:safe-enforced-quiescence}
Let $(S,s_0,\Sigma_I,\Sigma_O,\twoheadrightarrow)$ be a TIOLTS.
If $s\in S$ enables independent progress, then $\delta_{E}(s)\Rightarrow\delta_{S}(s)$.
\end{lemma}

\begin{proof}
We prove Lemma~\ref{lemma:safe-enforced-quiescence} by contradiction.
Let $(S,s_0,\Sigma_I,\Sigma_O,\twoheadrightarrow)$ be a TIOLTS.
If $s\in S$ enables independent progress, then $\forall d\in\Delta:s\xTwoheadrightarrow{d}$ or $\exists d\in\Delta,\exists o\in\Sigma_O:s\xtwoheadrightarrow{d}\xTwoheadrightarrow{o}$ (\cf{}~Definition~\ref{def:inp-en-ind-prog}).
Assume that it holds that $\delta_E(s)$ but not $\delta_S(s)$.
Then, $\forall o\in\Sigma_O,\forall d\in\Delta:s\not\xtwoheadrightarrow{(d,o)}$ and $\exists d\in\Delta:s\not\xtwoheadrightarrow{d}$.
However, this contradicts the assumption that $s$ enables independent progress.
Hence, it holds that $\delta_E(s)\Rightarrow\delta_S(s)$ if $s$ enables independent progress.
\end{proof}

We add $\delta_S$ and $\delta_E$ to $\out$ to distinguish both types of quiescence and
adjust $\tstraces$, accordingly.
This allows us to define \emph{live timed ioco} ($\mathltiocolts$) by extending
$\mathtiocodelta$ with outputs $\delta_S$ and $\delta_E$.
Hence, $\mathltiocolts$ not only guarantees output behaviors of implementation $\textit{im}$
to be \emph{safe} (\ie{} allowed to occur within the observed time interval as specified in $\textit{sp}$), but 
also requires $\textit{im}$ to be \emph{live} (\ie{} to progress with an output 
within a time interval if enforced by $\textit{sp}$).

\begin{definition}\label{def:ltioco-tiolts}
    Let $\textit{im}$, $\textit{sp}$ be TIOLTS over $\Sigma=\Sigma_I\cup\Sigma_O$, $s,s'\in S$, $S'\subseteq S$, 
    $\xi\in(\Delta\times(\Sigma\cup\{\delta_\gamma\}))$.
    \begin{itemize}
        \item $s\after\xi:=\{s'\mid s\xTwoheadrightarrow{\xi}s'\}$,
        \item $\outl(s):=\{(d,o)\mid d\in\Delta,o\in\Sigma_O,s\xTwoheadrightarrow{(d,o)}\}\cup\{\delta_\gamma\mid\delta_\gamma(s)\}$,
        \item $\outl(S'):=\bigcup_{s\in S'}\outl(s)$,
        \item $\tstracesl(s):=\{\xi\mid s\xTwoheadrightarrow{\xi}\}$, where $s'\xtwoheadrightarrow{\delta_\gamma}s'$ iff $\delta_\gamma(s')$, and
    \item $\textit{im}\,\mathltiocolts\,\textit{sp}:\Leftrightarrow
    \forall\xi\in\tstracesl(\textit{sp}):\outl(\textit{im}\after\xi)\subseteq\outl(\textit{sp}\after\xi)$
    \end{itemize}
\end{definition}

Obviously, using two different quiescence symbols does not increase complexity of conformance checking
as compared to $\mathtiocodelta$ in~Def.~\ref{def:tioco-auxiliary}.

\begin{example}
State $\angles{\ell_1,x=0}$ in Fig.~\ref{fig:tioco-flaws} is quiescent, whereas
$\angles{\ell_2,x=0}$ and $\langle\ell_3,\allowbreak x=0\rangle$ are not.
With our improved definition, $\angles{\ell_1,x=0}$ 
is enforced-quiescent, whereas
$\angles{\ell_2,x=0}$ and $\angles{\ell_3,x=0}$ are safe-quiescent.
States $\angles{\ell_4,x=0}$ and $\angles{\ell_5,x=0}$ are neither safe-quiescent 
nor enforced-quiescent due to the invariants of $\ell_4$ and $\ell_5$.
Hence, $\mathltiocolts$ is now able to reject $\ell_3$ as incorrect implementation of $\ell_4$ 
and $\ell_5$ as both $\ell_4$ and $\ell_5$ are not quiescent, whereas $\ell_3$ is safe-quiescent.
For all other cases, $\mathltiocolts$ yields the same results as listed in Fig.~\ref{fig:tioco-flaws}.
\end{example}

\begin{lemma}\label{lemma:ltioco-preorder}
$\mathltiocolts$ is a preorder on the set of input-enabled TIOLTS.
\end{lemma}

\begin{proof}
Let $p,q,r$ be input-enabled TIOLTS being derived from TIOA, and\linebreak $p\mathltiocolts q$ and $q\mathltiocolts r$.
It holds by Definition~\ref{def:ltioco-tiolts} that $p\mathltiocolts p$, \ie{} $\mathltiocolts$ is reflexive.

It remains to be shown that $p\mathltiocolts r$, \ie{} $\forall\xi\in\tstracesl(r):\outl(p\after\xi)\allowbreak\subseteq\outl(r\after\xi)$.
Let $\xi\in\tstracesl(r)$.
If $\xi\in\tstracesl(q)$, then $\forall\xi\in\tstracesl:\outl(p\after\xi)\subseteq\outl(r\after\xi)$ follows from transitivity of $\subseteq$.

The case of $\xi\notin\tstracesl(q)$ remains, \ie{} the case where behaviors are not present in $q$ such that $\xi\in\tstracesl(p)$, $\xi\notin\tstracesl(q)$, and $\xi\in\tstracesl(r)$.
We prove this part by contradiction.
Suppose, $\forall\xi\in\tstracesl(r):\outl(p\after\xi)\allowbreak\subseteq\outl(r\after\xi)$ fails for a $\xi\in\tstracesl(r)\setminus\tstracesl(q)$, \ie{} such a $\xi$ exists.
Trace $\xi$ decomposes into $\xi_1\cdot(d,a)\cdot\xi_2$ where $\xi_1\in\tstracesl(q)$ but $\xi_1\cdot(d,a)\notin\tstracesl(q)$.
Since $\outl(p\after\xi_1)\subseteq\outl(q\after\xi_1)$, $a\notin\Sigma_O\cup\{\delta_S,\delta_E\}$.
Additionally, $a\in\Sigma_I$ contradicts input-enabledness of $q$.
Thus, $\xi\in\tstracesl(q)$ and $\mathltiocolts$ is transitive.

From reflexivity and transitivity of $\mathltiocolts$ it follows that $\mathltiocolts$ is indeed a preorder on input-enabled TIOLTS.
\end{proof}

Furthermore, we can prove that $\mathltiocolts$ is \emph{sound} (\ie{} strictly more discriminating) 
with respect to $\mathtiocodelta$ in the 
sense that $\textit{im}\mathltiocolts \textit{sp}\Rightarrow \textit{im}\mathtiocodelta \textit{sp}$ (but not vice versa).

\begin{theorem}[Correctness of $\mathltiocolts$]\label{theorem:soundness-ltioco}
    Let $\textit{im}$ and $\textit{sp}$ be TIOLTS
    with $\textit{im}$ being input-enabled and enabling independent progress.
    \begin{itemize}
      \item $\textit{im}\,\mathltiocolts\,\textit{sp}\,\Rightarrow\,\textit{im}\,\mathtiocodelta\,\textit{sp}$
      \item $\textit{im}\,\mathltiocolts\,\textit{sp}\,\Rightarrow\,\textit{im}\,\mathtiocodelay\,\textit{sp}$
      \item $\textit{im}\,\mathtiocodelta\,\textit{sp}\,\Rightarrow\,\textit{im}\,\mathltiocolts\,\textit{sp}$ does, in general, \emph{not} hold.
    \end{itemize}
    Additionally, let $\textit{sp}$ also be input-enabled.
    \begin{itemize}
    \item $\textit{im}\,\mathltiocolts\,\textit{sp}\,\Rightarrow\,
          \traces^w(\textit{im})\,\subseteq\,\traces^w(\textit{sp})$
   \end{itemize}
\end{theorem}

\begin{proof}
Let \emph{im} and \emph{sp} be TIOLTS with \emph{im} being input-enabled and enabling independent progress.

\medskip

First, we prove $\textit{im}\,\mathltiocolts\,\textit{sp}\,\Rightarrow\,\textit{im}\,\mathtiocodelta\,\textit{sp}$.
The only difference between $\mathtiocodelta$ and $\mathltiocolts$ is $\delta_S$ because $\delta_E(\angles{\ell,u})\Leftrightarrow\delta(\angles{\ell,u})$, \ie{} enforced quiescence, coincides with classical quiescence.
When we remove the output symbol $\delta_S$ from the $\outl$ sets it holds that $\outl(\textit{im}\after\xi)\subseteq\outl(\textit{sp}\after\xi)\Rightarrow(\outl(\textit{im}\after\xi)\setminus\{\delta_S\})\subseteq(\outl(\textit{sp}\after\xi)\setminus\{\delta_S\})$.
Hence, $\textit{im}\mathltiocolts \textit{sp}\Rightarrow \textit{im}\mathtiocodelta\textit{sp}$.

\medskip

Next, we prove $\textit{im}\,\mathltiocolts\,\textit{sp}\,\Rightarrow\,\textit{im}\,\mathtiocodelay\,\textit{sp}$.
The difference between $\mathtiocodelay$ and $\mathltiocolts$ is the out-set, containing output and delays for $\mathtiocodelay$ and pairs of outputs and delays and quiescence ($\delta_E$ and $\delta_S$) for $\mathltiocolts$.
Note, that we do not have to consider the differences in the $\tstraces{}$ and $\tstracesl$, respectively, as these differences are already captured by the out-sets.

When only considering delays, $\textit{im}\,\mathtiocodelay\,\textit{sp}$ holds if \textit{im} does not allow more delays than \textit{sp}.
This behavior is captured by $\mathltiocolts$ as \textit{im} may only introduce an invariant, resulting in \textit{im} not having output symbol $\delta_S$, preserving the subset relation.
Making the invariant stricter is already captured as outputs are always pairs of delays and actions.
Furthermore, $\textit{im}\,{\not\mathtiocodelay}\,\textit{sp}$ if \textit{im} allows for more delays than \textit{sp}.
This is also captured by $\mathltiocolts$ as allowing more delays means removing the invariant of the corresponding location, thus introducing output symbol $\delta_S$.
Only changing the invariant to a greater value either violates independent progress (if no output action or $\tau$ is possible after this delay) or also allows outputs after these greater delays (which is covered through outputs in $\mathltiocolts$ being pairs of delays and actions).

\medskip

Next, we show that $\textit{im}\,\mathtiocodelta\,\textit{sp}\,\Rightarrow\,\textit{im}\,\mathltiocolts\,\textit{sp}$ does, in general, \emph{not} hold.
Figure~\ref{fig:tioco-flaws} provides an example.
Here, $\llbracket\mathcal{A}_3\rrbracket_S\mathtiocodelta\llbracket\mathcal{A}_4\rrbracket_S$ holds, but $\llbracket\mathcal{A}_3\rrbracket_S\mathltiocolts\allowbreak\llbracket\mathcal{A}_4\rrbracket_S$ does not hold.

\medskip

Now, let \emph{im} and \emph{sp} be TIOLTS with \emph{im} and \emph{sp} being input-enabled and \emph{im} enabling independent progress.
Finally, we prove that $\textit{im}\,\mathltiocolts\,\textit{sp}\,\Rightarrow\, \traces^w(\textit{im})\,\allowbreak\subseteq\,\traces^w(\textit{sp})$.
From $\textit{im}\,\mathltiocolts\,\textit{sp}$ it follows by definition that $\forall\xi\in\tstracesl(\textit{sp}):\outl(\textit{im}\after\xi)\subseteq\outl(\textit{sp}\allowbreak\after\xi)$, resulting in trace inclusion for $\tstracesl$.
For $\textit{im}\,\mathltiocolts\,\textit{sp}\,\Rightarrow\, \traces^w(\textit{im})\,\subseteq\,\traces^w(\textit{sp})$, we have to show that removing quiescence symbols $\delta_E$ and $\delta_S$ from all $\tstracesl$ (resulting in $\traces^w$) preserves the subset relation.
In $\tstracesl$, $\delta_E$ and $\delta_S$ are added with self-loops to the respective TIOLTS states.
Therefore, we remove all $\xi\in\tstracesl(\textit{sp})$ and $\xi\in\tstracesl(\textit{im})$ containing $\delta_E$ and/or $\delta_S$.
By $\Xi_\textit{sp}^\delta$ and $\Xi_\textit{im}^\delta$ we denote the sets of traces containing $\delta_E$ and $\delta_S$.
Due to $\textit{im}\,\mathltiocolts\,\textit{sp}$ it holds that $\Xi_\textit{im}^\delta\subseteq\Xi_\textit{sp}^\delta$.
Hence, removing all $\xi\in\Xi_\textit{im}^\delta$ from $\tstracesl(\textit{sp})$ and $\tstracesl(\textit{im})$ does not effect the subset relation.
Furthermore removing all $\xi\in\Xi_\textit{sp}^\delta\setminus\Xi_\textit{im}^\delta$ also does not effect the subset relation as $\forall\xi\in\Xi_\textit{sp}^\delta\setminus\Xi_\textit{im}^\delta:(\xi\in\tstracesl(\textit{sp})\land\xi\notin\tstracesl(\textit{im}))$.
Finally, we have to require input-enabledness for \emph{sp} such that $\mathltiocolts$ is a preorder (\cf{}~Lemma~\ref{lemma:ltioco-preorder}).
We have this requirement as $\traces^w(\textit{im})\,\subseteq\,\traces^w(\textit{sp})$ also is a preorder.
Therefore, $\textit{im}\,\mathltiocolts\,\textit{sp}\,\Rightarrow\, \traces^w(\textit{im})\,\subseteq\,\traces^w(\textit{sp})$.
\end{proof}

Note, that $\textit{im}\,\mathtiocodelay\,\textit{sp}\,\Rightarrow\,\textit{im}\,\mathltiocolts\,\textit{sp}$ 
does not hold as $\mathtiocodelay$ has no notion of quiescence, and 
$\textit{im}\,\mathltiocolts\,\textit{sp}\,\Rightarrow\,\traces^s(\textit{im})\,\subseteq\,\traces^s(\textit{sp})$ 
does not hold as $\mathltiocolts$ is limited to observable (weak) steps of timed (suspension) traces.

\subsection{Compositionality}

For investigating compositionality
of $\mathltiocolts$, we first define parallel composition
of TIOA also at the level of TIOLTS.

\begin{definition}[TIOLTS Composition]\label{def:tiolts-compo-epsi}
Let $(S_j,s_{0_j},\Sigma_{I_{j}},\Sigma_{O_{j}},\twoheadrightarrow_j)$ with $j\in\{1,2\}$ 
be TIOLTS of composable TIOA.
The \emph{parallel product} is a TIOLTS
$(S_1 \times S_2, (s_{0_1},s_{0_2}), \Sigma_{I_{1\parallel 2}}, \Sigma_{O_{1\parallel 2}},\twoheadrightarrow_{1 \parallel 2})$, 
where $\Sigma_{I_{1\parallel 2}}$ and $\Sigma_{O_{1\parallel 2}}$ are defined according to
Def.~\ref{def:tioa-compo-epsi} and 
$\twoheadrightarrow_{1 \parallel 2}$ is the least relation satisfying the rules:
\begin{tabbing}
  (5)\quad \= $(s_1,s_2)\xtwoheadrightarrow{\sigma}_{1 \parallel 2} (s_1',s_2)$ \= if \= \kill
  (1) \> $(s_1,s_2)\xtwoheadrightarrow{\sigma}_{1 \parallel 2} (s_1',s_2)$ \> if \> $s_1\xtwoheadrightarrow{\sigma}_1 s_1'$, $\sigma\in(\Sigma_1\setminus\Sigma_2)\cup\{\tau\}$, \\
  (2) \> $(s_1,s_2)\xtwoheadrightarrow{\sigma}_{1 \parallel 2} (s_1,s_2')$ \> if \> $s_2'\xtwoheadrightarrow{\sigma}_2 s_2'$, $\sigma\in(\Sigma_2\setminus\Sigma_1)\cup\{\tau\}$, \\
  (3) \> $(s_1,s_2)\xtwoheadrightarrow{\tau}_{1 \parallel 2} (s_1',s_2')$ \> if \> $s_1\xtwoheadrightarrow{\sigma}_1 s_1'$, $s_2'\xtwoheadrightarrow{\sigma}_2 s_2'$ and $\sigma\in(\Sigma_1\cap\Sigma_2)$, and \\
  (4) \> $(s_1,s_2)\xtwoheadrightarrow{d}_{1 \parallel 2} (s_1',s_2')$ \> if \> $s_1\xtwoheadrightarrow{d}_1 s_1'$, $s_2'\xtwoheadrightarrow{d}_2 s_2'$ and $d\in\Delta$. \\
\end{tabbing}
\end{definition}
Rules~(1) and~(2) preserve transitions of non-shared (\ie{} unsynchronized) 
actions from both TIOLTS, whereas rule~(3) introduces silent transitions
for input/output action pairs synchronized between both TIOLTS.
Rule~(4) preserves (synchronous) delay steps of length $d$ enabled by both TIOLTS.
Rule~(5) handles inputs leading to the failure state 
in one of the components, where our notion of composable TIOA ensures that 
those actions leading to the failure state are not shared.
We conclude the following properties.

\begin{lemma}\label{lemma:tiolts-composition-properties}
Let $\mathcal{A}_{1}$ and $\mathcal{A}_{2}$ be composable TIOA.
\begin{enumerate}
  \item
  $\traces(\llbracket \mathcal{A}_{1\parallel 2}\rrbracket_{S}) = 
  \traces(\llbracket \mathcal{A}_{1}\rrbracket_{S}\parallel\llbracket\mathcal{A}_{2}\rrbracket_{S})$, and
  \item if $\mathcal{A}_{1}$ and $\mathcal{A}_{2}$ are \emph{input-enabled} and enable \emph{independent progress}, then this also holds for $\mathcal{A}_{1\parallel 2}$.
\end{enumerate}
\end{lemma}

\begin{proof}
Let $\mathcal{A}_{1}$ and $\mathcal{A}_{2}$ be composable TIOA.
We prove~(1) and~(2) separately.

\medskip

(1) Let $P=(S_P,s_0^P,\Sigma_I^P,\Sigma_O^P,\twoheadrightarrow_P)=\llbracket \mathcal{A}_{1\parallel 2}\rrbracket_{S}$ and $Q=(S_Q,s_0^Q,\Sigma_I^Q,\Sigma_O^Q,\allowbreak\twoheadrightarrow_Q\nobreak)=\llbracket \mathcal{A}_{1}\rrbracket_{S}\parallel\llbracket\mathcal{A}_{2}\rrbracket_{S}$. In order to prove $\traces(\llbracket \mathcal{A}_{1\parallel 2}\rrbracket_{S}) = \traces(\llbracket \mathcal{A}_{1}\rrbracket_{S}\parallel\llbracket\mathcal{A}_{2}\rrbracket_{S})$, we show the following:
\begin{itemize}
    \item $S_P=S_Q$.
    When deriving a TIOLTS from a TIOA, the set of states can only be reduced by location invariants.
    When composing two locations, their invariants are, by definition, conjugated.
    Therefore, the set of states of $\llbracket \mathcal{A}_{1\parallel 2}\rrbracket_{S}$ is determined by conjunction of location invariants of both $\mathcal{A}_1$ and $\mathcal{A}_2$.
    Furthermore, delay transitions only remain after composition if both $\llbracket \mathcal{A}_{1}\rrbracket_{S}$ and $\llbracket \mathcal{A}_{2}\rrbracket_{S}$ are able to perform a delay (\cf{}~Rule~(4) of Definition~\ref{def:tiolts-compo-epsi}).
    As these delay transitions are a result of location invariants, it holds that $S_P=S_Q$.
    \item $\Sigma_I^P=\Sigma_I^Q$ and $\Sigma_O^P=\Sigma_O^Q$.
    These equalities hold by definition (\cf{}~Definitions~\ref{def:tioa-compo-epsi} and~\ref{def:tiolts-compo-epsi}).
    \item $\twoheadrightarrow_P=\twoheadrightarrow_Q$.
    Similar to $S_P=S_Q$, TIOLTS transitions are dependent on on clock constraints, and additionally they depend on TIOA switches.
    As with $S_P=S_Q$, clock constraints are, by definition, conjugated.
    Hence, the set of transitions of $\llbracket \mathcal{A}_{1\parallel 2}\rrbracket_{S}$ is determined by conjunction of clock constraints of both $\mathcal{A}_1$ and $\mathcal{A}_2$, and, as with $S_P=S_Q$, it holds that $\twoheadrightarrow_P=\twoheadrightarrow_Q$.
\end{itemize}
Hence, it holds that $\traces(\llbracket \mathcal{A}_{1\parallel 2}\rrbracket_{S}) = \traces(\llbracket \mathcal{A}_{1}\rrbracket_{S}\parallel\llbracket\mathcal{A}_{2}\rrbracket_{S})$ as the sets of states, actions, and transitions are equal.

\medskip

(2) Let $\Sigma_I^1$ be the inputs of $\mathcal{A}_1$, $\Sigma_I^2$ be the inputs of $\mathcal{A}_2$, and $\Sigma_I^{1\parallel 2}=(\Sigma_I^1\cup\Sigma_I^2)\setminus(\Sigma_O^1\cup\Sigma_O^2)$ be the inputs of $\mathcal{A}_{1\parallel 2}$.
Rules~(1) and~(2) of TIOA composition ensure that inputs $\Sigma_I^1\setminus(\Sigma_I^2\cup\Sigma_O^2)$ and $\Sigma_I^2\setminus(\Sigma_I^1\cup\Sigma_O^1)$ are preserved, respectively.
As $\Sigma_I^1\cap\Sigma_I^2=\emptyset$, $\Sigma_I^{1\parallel 2}$ does not contain further inputs.
Therefore, input-enabledness is preserved under TIOA composition.
Furthermore, assume that TIOA composition does not preserve independent progress.
Hence, there is a restriction in $\mathcal{A}_{2}$ such that it holds for a state $s$ of $\llbracket\mathcal{A}_{1}\rrbracket_S$ (or vice versa) that $\exists d\in\Delta:s\not\xTwoheadrightarrow{d}$ or $\nexists d\in\Delta:s\xtwoheadrightarrow{d}\xTwoheadrightarrow{o}$ for an $o\in\Sigma_O^1$.
However, if such a restriction would exist, then the corresponding state in $\llbracket\mathcal{A}_{2}\rrbracket_S$ would enable independent progress as $\mathcal{A}_{2}$ enables independent progress.
Futhermore, output $o\in\Sigma_O^1$ (or $\Sigma_O^2$, respectively) does not obstruct independent progress if $o\in\Sigma_I^1\cap\Sigma_O^2$ or $o\in\Sigma_I^2\cap\Sigma_O^1$ is a common action as the matching input is always available due to input-enabledness.
The result is an internal action $\tau$ not obstructing independent progress.
Hence, TIOA composition preserves independent progress.
\end{proof}

Property~(1) ensures parallel composition on TIOA and TIOLTS
to commute with respect to timed-traces semantics such that 
a composed specification can be effectively built from 
the (finite) TIOA representations of its components.
Property~(2) ensures that input-enabled and independent-progress enabling TIOA are closed under parallel composition.
We now prove compositionality
of $\mathltiocolts$.

\begin{theorem}\label{theorem:compositionality}
Let $\textit{im}_{1}$, $\textit{im}_{2}$, $\textit{sp}_{1}$, and $\textit{sp}_{2}$ 
be \emph{input-enabled} and \emph{independent progress} enabling
TIOLTS of composable TIOA.
Then it holds that
$$(\textit{im}_{1}\,\mathltiocolts\,\textit{sp}_{1})\,\wedge\,
  (\textit{im}_{2}\,\mathltiocolts\,\textit{sp}_{2})\,\Rightarrow\,
  (\textit{im}_{1}\,\parallel\,\textit{im}_{2})\,\mathltiocolts\,
  (\textit{sp}_{1}\,\parallel\,\textit{sp}_{2}).$$
\end{theorem}

\begin{proof}
Let $\textit{im}_{1}$ and $\textit{im}_{2}$ as well as $\textit{sp}_{1}$ and $\textit{sp}_{2}$ be \emph{input-enabled} and \emph{independent progress} enabling TIOLTS of composable TIOA.
Additionally, it holds that $\textit{im}_{1}\mathltiocolts\textit{sp}_{1}$ and $\textit{im}_{2}\mathltiocolts\textit{sp}_{2}$.
In order to prove $\textit{im}_{1}\parallel\textit{im}_{2}\mathltiocolts\textit{sp}_{1}\parallel\textit{sp}_{2}$, we have to prove that $\forall\xi\in\tstracesl(\textit{sp}_1\parallel\textit{sp}_{2}):\outl(\textit{im}_1\parallel\textit{im}_{2}\after\xi)\subseteq\outl(\textit{sp}_1\parallel\textit{sp}_{2}\after\xi)$.
To prove this we first assume that Rule~(3) of TIOLTS composition (\cf{}~Definition~\ref{def:tiolts-compo-epsi}) results in becoming the respective output action instead of an internal action $\tau$ and prove $\textit{im}_{1}\parallel\textit{im}_{2}\mathltiocolts\textit{sp}_{1}\parallel\textit{sp}_{2}$ for this adjusted composition operator.
Afterwards, we \emph{hide} the output actions being generated by adjusted Rule~(3) by replacing them with internal actions $\tau$ such that we prove Theorem~\ref{theorem:compositionality} for TIOLTS composition as defined in Definition~\ref{def:tiolts-compo-epsi}.

Let $\mu\in\outl(\textit{im}_1\parallel\textit{im}_{2}\after\xi)$ such that, \Wlog{}, $\mu\in\outl(\textit{im}_1\after\xi)$ with $\outl(\textit{im}_1\parallel\textit{im}_{2}\after\xi)\subseteq\Sigma_O$.
Then, $\mu\in\outl(\textit{sp}_1\parallel\textit{sp}_{2}\after\xi)$ as otherwise $\textit{im}_1$ would have more output behavior than $\textit{sp}_1$ such that $\textit{im}_{1}\mathltiocolts\textit{sp}_{1}$ would not hold.
Next, assume that $\delta_E\in\outl(\textit{im}_1\after\xi)$.
Then, it also holds that $\delta_E\in\outl(\textit{im}_1\parallel\textit{im}_{2}\after\xi)$ if $\nexists\mu\in\Sigma_O:\mu\in\outl(\textit{im}_2\after\xi)$.
Otherwise, it also holds that $\delta_E\notin\outl(\textit{im}_2\after\xi)$ such that $\delta_E\notin\outl(\textit{sp}_1\parallel\textit{sp}_{2}\after\xi)$.
The reasoning for $\delta_S$ is analogous.
Hence, $\textit{im}_{1}\parallel\textit{im}_{2}\mathltiocolts\textit{sp}_{1}\parallel\textit{sp}_{2}$ with the adjusted Rule~(3) as described above.

Next, we replace the adjusted Rule~(3) by the original one to prove Theorem~\ref{theorem:compositionality}.
Here, $\textit{im}'$ and $\textit{sp}'$ describe the adjusted variants of $\textit{im}_1\parallel\textit{im}_{2}$ and $\textit{sp}_1\parallel\textit{sp}_{2}$ where outputs of Rule~(3) are hidden, \ie{} replaced by $\tau$.
Additionally, let $\xi'\in\tstracesl(\textit{sp}')$ denote the tstrace corresponding to $\xi\in\tstracesl(\textit{sp}_1\parallel\textit{sp}_{2})$.
Assume, Theorem~\ref{theorem:compositionality} does not hold.
Then, there exists a $\mu\neq\tau$ such that $\mu\in\outl(\textit{im}_1\parallel\textit{im}_{2}\after\xi)$, $\mu\in\outl(\textit{sp}_1\parallel\textit{sp}_{2}\after\xi)$, $\mu\in\outl(\textit{im}'\after\xi')$, and $\mu\notin\outl(\textit{sp}'\after\xi')$.
However, $\textit{sp}_1\parallel\textit{sp}_{2}$ is input-enabled because $\textit{sp}_1$ and $\textit{sp}_2$ are input-enabled (\cf{}~Lemma~\ref{lemma:tiolts-composition-properties}).
Therefore, we impose that $\textit{im}_1\parallel\textit{im}_{2}$ implements every input $i\in\Sigma_I$ for every state $s\in S$ such that $s_{\textit{sp}_1\parallel\textit{sp}_{2}}\xrightarrow{i}s_{\textit{sp}_1\parallel\textit{sp}_{2}}'\Rightarrow s_{\textit{im}_1\parallel\textit{im}_{2}}\xrightarrow{i}s_{\textit{im}_1\parallel\textit{im}_{2}}'$.
This means, we dictate how $\textit{im}_1\parallel\textit{im}_{2}$ should behave after $\xi$.
Therefore, $\textit{im}'$ cannot have any additional output behaviors after $\xi'$ not being in $(\textit{sp}'\after\xi')$.
Hence, $\textit{im}_{1}\parallel\textit{im}_{2}\mathltiocolts\textit{sp}_{1}\parallel\textit{sp}_{2}$ and Theorem~\ref{theorem:compositionality} is correct.
\end{proof}

\subsection{Symbolic Live Timed Input/Output Conformance Testing}

Concerning the practical intractability of infinitely branching TIOLTS, 
\emph{zone graphs} have been proposed as \emph{finite}
representation of TA semantics~\cite{Dill1990}.
A zone graph $(\mathcal{Z},\rightsquigarrow)$
of TIOA $\mathcal{A}$ consists of a \emph{transition relation} $\rightsquigarrow$
on a set $\mathcal{Z}$ of \emph{symbolic states} by means of pairs $\angles{\ell,\varphi}$ of 
locations $\ell\in L$ and \emph{zones} $\varphi\in\mathcal{B(C)}$.
A zone represents a (potentially infinite) maximum set $D$ of clock valuations satisfying 
clock constraint $\varphi$, where we assume zones in \emph{canonical form} 
by requiring $D$ to be \emph{closed under entailment} (\ie{} $\varphi$ cannot be strengthened
without changing $D$).
We may write $D$ as a synonym for $\varphi$ and 
use the notations $D^\uparrow=\{u+d\mid u\in D,d\in\mathbb{T}\}$ and $R(D)=\{[R\mapsto 0]u\mid u\in D\}$.
Although zone graphs $(\mathcal{Z},\rightsquigarrow)$ are, again, not necessarily finite, 
an equivalent, finite zone graph $(\mathcal{Z},\rightsquigarrow_k)$ can be obtained
with $\rightsquigarrow_k$,
(1) by constructing an equivalent \emph{diagonal-free} TA only
containing atomic clock constraints of the form $x\sim r$~\cite{Berard1998}, 
and (2) by constructing for this TA a \emph{$k$-bounded} 
zone graph with all zones being bound 
by a maximum global \emph{clock ceiling} $k$ 
using $k$-normalization~\cite{Rokicki1994,Pettersson1999}.
Here, the basic idea of $k$-normalization is to set the value of $k$ to the greatest constant appearing in any clock constraint in the TA.
Then, we replace each difference constraint by a difference greater than $k$ 
(\ie{} a difference constraint stating that the difference is greater than $k$).

As zone-graph constructions from TA ignore 
switch labels, they are likewise applicable to TIOA.
However, in order to lift $\mathltiocolts$ to
zone graphs of specifications $\mathcal{A}_{\textit{sp}}$ and implementations
$\mathcal{A}_{\textit{im}}$ given as TIOA, actions related to
TIOA switches (including $\tau$) must be also included as labels for the respective 
transitions between the corresponding symbolic states.
In contrast, symbolic transitions not corresponding to switches of the TIOA
are labeled with the special void symbol $\epsilon\notin\Sigma$.
We define input/output-labeled zone graph (IOLZG) representations of TIOA as follows.

\begin{definition}[IOLZG]\label{def:labeled-zone-graph}
    An \emph{IOLZG} of TIOA $\mathcal{A}=(L,\ell_0,\Sigma_I,\Sigma_O,\rightarrow,I)$ is a tuple $(\mathcal{Z},z_0,\Sigma_I,\Sigma_O,\rightsquigarrow\nobreak)$, where
    \begin{itemize}
        \item $\mathcal{Z}=L\times\mathcal{B(C)}$ is a set of \emph{symbolic states} with \emph{initial state} $z_0=\angles{\ell_{0},D_{0}}\in \mathcal{Z}$,
        \item $\Sigma_{\tau}=\Sigma_I\cup\Sigma_O\cup\{\tau\}$ is a set of \emph{labels}, and
        \item ${\rightsquigarrow}\subseteq\mathcal{Z}\times(\Sigma_{\tau}\cup \{\epsilon\})\times\mathcal{Z}$ 
        is a \emph{symbolic transition relation} being the least relation satisfying the following rules:
    \begin{itemize}
        \item $\angles{\ell,D}\xrightsquigarrow{\epsilon}\angles{\ell,D^\uparrow\land I(\ell)}$ and
        \item $\angles{\ell,D}\xrightsquigarrow{\sigma}\angles{\ell',R(D\land g)\land I(\ell')}$ if $\ell\xrightarrow{g,\sigma,R}\ell'$.
    \end{itemize}
    \end{itemize}
    Let $\angles{\ell,D}\in\mathcal{Z}$ be a symbolic state.
    We further use the following notations.
    \begin{itemize}
      \item $\angles{\ell,D}\xrightsquigarrow{d}\angles{\ell',R(D\land g)\land I(\ell')}$ if $\exists u\in D:u\in g\land([R\mapsto0](u+d))\in R(D\land g)\land I(\ell')$,
      \item $\angles{\ell,D}\xrightsquigarrow{(d,\sigma)}$ if $\exists\angles{\ell'',D''}\in\mathcal{Z}:\angles{\ell,D}\xrightsquigarrow{d}\angles{\ell'',D''}\xrightsquigarrow{\sigma}\angles{\ell',D'}$,
      \item $\angles{\ell,D}\xrightsquigarrow{(d_1,\sigma_1)\cdots(d_n,\sigma_n)}$ if $\exists\angles{\ell_1,D_1},\ldots,\angles{\ell_n,D_n}\in\mathcal{Z}:\angles{\ell,D}\xrightsquigarrow{(d_1,\sigma_1)}\angles{\ell_1,D_1}\allowbreak\xrightsquigarrow{(d_2,\sigma_2)}\ldots\xrightsquigarrow{(d_n,\sigma_n)}\angles{\ell_n,D_n}$ with $n\in\mathbb{N}_0$,
      \item $\angles{\ell,D}$ is \emph{input-enabled} 
      iff $\forall i\in \Sigma_I,\forall d\in D:\exists\angles{\ell',D'}\in\mathcal{Z}:\angles{\ell,D}\xrightsquigarrow{(d,i)}\angles{\ell',D'}\land d\in D'$, and
      \item $\angles{\ell,D}$ enables \emph{independent progress} 
      iff $(\forall d\in\Delta:\angles{\ell,D}\xrightarrow{d})$ 
      or $\exists d\in\Delta,\exists o\in\Sigma_O:\angles{\ell,D}\xrightarrow{d}\xrightarrow{o}$.
    \end{itemize}
\end{definition}
An IOLZG is \emph{input-enabled} and enables \emph{independent progress} if all its state do.
Again, we obtain weak steps by replacing $\rightsquigarrow$ by $\Rightsquigarrow$, where
in both relations, $\epsilon$-steps are treated as unobservable.
By $\llbracket\mathcal{A}\rrbracket_{\mathcal{Z}}^{x}$, $x\in\{w,s\}$, we refer
to the weak/strong IOLZG of TIOA $\mathcal{A}$, again, by possibly omitting $x$.
In fact, $k$-normalization also applies to IOLZG, where switch labels may cause
duplications of transitions but, however, do not affect the set of symbolic states.
Hence, the correctness claim for zone graphs of 
TA (\cf{}~\cite{Bengtsson2004}) also holds for IOLZG of TIOA.

\begin{theorem}\label{theorem:labeled-zone-graph-operational-semantics}
    Let $s_0 = \angles{\ell_0,u_0}$ be the initial state 
    of TIOLTS $\llbracket \mathcal{A}\rrbracket_{S}$ of TIOA $\mathcal{A}$
    and $\angles{\ell_,\{u_0\}}$ be the initial state of IOLZG $\llbracket\mathcal{A}\rrbracket_\mathcal{Z}$.
    \begin{itemize}
        \item (Soundness) $\angles{\ell_0,\{u_0\}}\xrightsquigarrow{\xi}_k\angles{\ell,D}$ 
        implies $\angles{\ell_0,u_0}\xtwoheadrightarrow{\xi}\angles{\ell,u}$ for all $u\in D$.
        \item (Completeness) $\angles{\ell_0,u_0}\xtwoheadrightarrow{\xi}\angles{l,u}$ 
        implies $\angles{\ell_0,\{u_0\}}\xrightsquigarrow{\xi}_k\angles{\ell,D}$ such that $u\in D$.
    \end{itemize}
\end{theorem}

\begin{proof}
The correctness of this proof directly follows from correctness of $k$\hyp{}normalization~\cite{Bengtsson2004} and the fact that labeled zone graphs connect the same symbolic states through transitions as zone graphs with the only difference being the labels of the labeled zone graphs (\cf{}~\cite{Bengtsson2004} and Definition~\ref{def:labeled-zone-graph}).
Furthermore, adding labels $\epsilon$ to transitions not corresponding to TIOA switches does not obstruct this result as these transitions are only use to apply operation $D^\uparrow$.
\end{proof}

\begin{example}
Figure~\ref{fig:labeled-zone-graph} shows an extract from the ($k$-normalized) 
IOLZG of the TIOA in Fig.~\ref{fig:vending-machine-ta}. 
Here, $k=20$ is the largest constant appearing in all clock constraints
such that every value of clocks $x$ larger than 20 falls into zone $x>20$.
The initial zone restricts all clock values to 0.
Symbolic state $\angles{\text{idle},x\leq20,x=y}$ comprises all TIOLTS states
being in location \emph{idle} as long as $x\leq20$ holds, and, similarly, for 
the symbolic states with location \emph{off}.
On reaching location \emph{as} (\emph{add sugar}), all clocks are reset.
Symbolic state $\angles{\text{as},x\leq10,y\leq20,y-x\geq10}$ thus aggregates all
clock constraints of related TIOLTS runs.
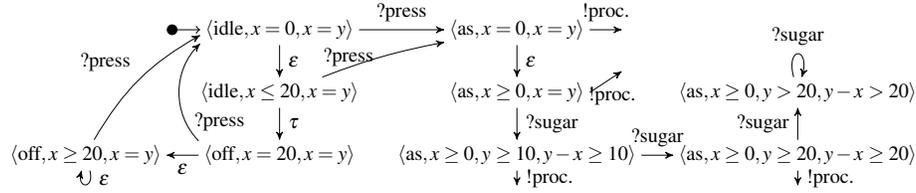
\begin{figure}[tp]
    \centering

\begin{adjustbox}{max width=\textwidth}
\scalebox{\scalefactor}{
\begin{tikzpicture}[node distance=.5]

\node[labeledstate, initial, initial where=left] (idle0) {$\angles{\text{idle},x=0,x=y}$};
\node[labeledstate, below=of idle0] (idle20leq) {$\angles{\text{idle},x\leq20,x=y}$};
\node[labeledstate, below=of idle20leq] (off20) {$\angles{\text{off},x=20,x=y}$};
\node[labeledstate, left=of off20] (off20geq) {$\angles{\text{off},x\geq20,x=y}$};

\node[labeledstate, right=of off20] (as010) {$\angles{\text{as},x\geq0,y\geq10,y-x\geq10}$};
\node[labeledstate, above=of as010] (as0geq) {$\angles{\text{as},x\geq0,x=y}$};
\node[labeledstate, above=of as0geq] (as0) {$\angles{\text{as},x=0,x=y}$};
\node[labeledstate, right=of as010] (as020) {$\angles{\text{as},x\geq0,y\geq20,y-x\geq20}$};
\node[labeledstate, above=of as020] (as020k) {$\angles{\text{as},x\geq0,y>20,y-x>20}$};

\node[labeledstate, right=of as0] (as0right) {};
\node[labeledstate, right=of as0geq, yshift=1em] (as0geqright) {};
\node[labeledstate, below=of as010, yshift=.7em] (as010below) {};
\node[labeledstate, below=of as020, yshift=.7em] (as020below) {};

\draw[live] (idle0) to node[auto] {$\epsilon$} (idle20leq);
\draw[live] (idle20leq) to node[auto] {$\tau$} (off20);
\draw[live, bend left=45] (off20.170) to node[pos=.21, right] {?press} (idle0.190);
\draw[live] (off20) to node[auto] {$\epsilon$} (off20geq);
\draw[live, bend left=15] (off20geq) to node[auto] {?press} (idle0.184);
\draw[live, loop below, looseness=6] (off20geq) to node[pos=.2, right] {$\epsilon$} (off20geq);
\draw[live] (idle0) to node[auto] {?press} (as0);
\draw[live, bend left=5] (idle20leq) to node[pos=.5, left] {?press} (as0);

\draw[live] (as0) to node[auto] {$\epsilon$} (as0geq);
\draw[live] (as0geq) to node[auto] {?sugar} (as010);
\draw[live] (as010) to node[auto] {?sugar} (as020);
\draw[live] (as020) to node[auto] {?sugar} (as020k);
\draw[live, loop above] (as020k) to node[auto] {?sugar} (as020k);
\draw[live] (as010) to node[auto] {!proc.} (as010below);
\draw[live] (as020) to node[auto] {!proc.} (as020below);
\draw[live] (as0) to node[auto] {!proc.} (as0right);
\draw[live] (as0geq.east) to node[pos=.7, below] {!proc.} (as0geqright);

\end{tikzpicture}
}
\end{adjustbox}
    \caption{Example for a $k$-Normalized IOLZG}\label{fig:labeled-zone-graph}
\end{figure}
\end{example}
As all TIOLTS states comprised in a symbolic state share the same 
visible behaviors (up to different clock valuations),
IOLZG can be used as a basis for checking $\mathltiocolts$ between respective TIOA.
In particular, if a zone of a symbolic state is downward-closed, outputs of that 
state are enforced as runs may not starve in that state.
Correspondingly, we can lift all auxiliary definitions of $\mathltiocolts$ from TIOLTS to 
IOLZG (marked by index $\mathcal{Z}$).
For $\outz$, we have to check for a given symbolic state reached by some 
\emph{tstrace} whether it is possible to extend the \emph{tstrace} by an output 
of that symbolic state such that the resulting extended \emph{tstrace} is still a valid \emph{tstrace}.
For instance, assume a simple IOLZG with $\angles{\ell_0,x\geq5}\xrightsquigarrow{!o}\angles{\ell_1,x<3}$:
state $\angles{\ell_0,x\geq5}$ has output $o$ which is only
enabled as long as $x<3$ holds as the state reached by that output is $\angles{\ell_1,x<3}$.
As the set of all valid extensions of \emph{tstraces}
by means of pairs of delays and subsequently enabled output actions
of one symbolic state is, in general, infinite, they do 
not provide a reasonable basis for effectively
checking $\mathltiocolts$ on zone-graph representations of TIOA.
However, a symbolic solution (\ie{} comparing the 
timing constraints for output-action occurrences of symbolic states)
is also not feasible for checking $\mathltiocolts$
due to the (generally) unrelated names of locations and clocks of the
two different TIOA under consideration.
To solve this problem, we instead employ the notion of \emph{spans}~\cite{Guha2012}: 
the span of clock $c$ in zone $D$ denotes the minimum time interval
containing the minimum and maximum valuations of $c$ enabled in $D$.
We use $\infty$ to 
denote upward-open intervals (\ie{} $d < \infty$ for all $d \in \mathbb{T})$.

\begin{definition}[Span]\label{def:span}
Let $D$ be a zone and $c\in C$.
\begin{itemize}
    \item $\emph{span}(c,D)=(\textit{lo},\textit{up})\in\mathbb{T}_C\times(\mathbb{T}_C\cup\{\infty\})$ is the minimal interval \st{} $\forall u\in D: u(c)\geq\textit{lo}\land u(c)\leq\textit{up}$.
  \item $(\textit{lo},\textit{up})\preceq(\textit{lo}',\textit{up}')
  \Leftrightarrow\textit{lo}\geq\textit{lo}'\land\textit{up}\leq\textit{up}'$.
  \item $\emph{span}(D)=(\textit{lo},\textit{up})\Leftrightarrow \forall c\in C:
  (\textit{lo},\textit{up})\preceq\emph{span}(c,D) \wedge \exists c',c''\in C: \emph{span}(c',D)=(lo,up') \wedge\linebreak \emph{span}(c'',D)=(lo'',up)$.
\end{itemize}
\end{definition}
Given a span $\textit{sp}=(\textit{lo},\textit{up})$, we write $d\in \textit{sp}$ for short 
if $d\geq \textit{lo}$ and $d\leq \textit{up}$ hold.
Based on the notion of spans, we are able define
\emph{span traces} $(\textit{sp}_1,\sigma_1),\ldots,(\textit{sp}_n,\allowbreak \sigma_n)$
as sequences of pairs of spans
and action occurrences denoting (maximum) sets
of all valid timed traces 
$(d_1,\sigma_1),\ldots,(d_n,\allowbreak \sigma_n)$
of a given TIOA with equal untimed traces $\sigma_1,\ldots,\sigma_n$
and $d_i \in sp_i$ for $1\leq i\leq n$.

\begin{example}
A span trace of the TIOA in Fig.~\ref{fig:vending-machine-ta}
is, for instance, given as
$\textit{spt}=((20,\infty),\textit{?press}),((0,20),\allowbreak\textit{?press}),((10,\infty),\textit{?sugar})$.
This span trace comprises all timed traces that
first perform the invisible $\tau$-switch leading to location $\emph{off}$
after exactly 20 time units.
The first visible step, performing output action $\textit{?press}$,
then corresponds to the switch leading from location 
$\emph{off}$ back to location $\emph{idle}$ after at least 20 time units
(due to the constraint of the $\tau$-switch).
The second occurrence of output action $\textit{?press}$ corresponding to the switch
leading from location $\emph{idle}$ to location $\emph{add sugar}$ has
to be performed at least 0 and at most 20 time units after the previous step.
Afterwards, for the self-switch of location $\emph{add sugar}$ labeled $\emph{?sugar}$
to be enabled, at least 10 time units must elapse.
\end{example}
Please note that the set of valid timed traces of a given
untimed trace may not be representable by a single
span trace (\eg{} in case of non-deterministic TIOA).
The minimal, yet complete set of span traces
comprising all valid timed traces of a given TIOA $\mathcal{A}$
can be defined with respect to the corresponding 
IOLZG representation of $\mathcal{A}$ as follows.

\begin{definition}[Span Trace]\label{def:span-trace}
Let $\mathcal{A}=(L,\ell_0,\Sigma_I,\Sigma_O,\rightarrow,I)$ be a TIOA 
with IOLZG $(\mathcal{Z},z_0,\Sigma_I,\Sigma_O,\rightsquigarrow)$.
By $\Psi_\mathcal{Z}$ we denote the set of \emph{span traces} 
of $\mathcal{A}$ being the \emph{least set} such that
$(\textit{sp}_1,\sigma_1),\ldots,(\textit{sp}_n,\sigma_n)\in\Psi_\mathcal{Z}\Leftrightarrow 
z_0 \xrightsquigarrow{(d_1,\sigma_1)\cdots(d_n,\sigma_n)}$, where $d_i\in\textit{sp}_i,1\leq i\leq n$.
\end{definition}
We can show that the set of span traces derived from the IOLZG
representation of a TIOA exactly comprises the
set of timed traces of the respective TIOLTS representation of the TIOA.

\begin{lemma}\label{lemma:span-trace-correctness}
Let $\mathcal{A}=(L,\ell_0,\Sigma_I,\Sigma_O,\rightarrow,I)$ be a TIOA
with TIOLTS $(S,s_0,\Sigma_I,\allowbreak\Sigma_O,\allowbreak\twoheadrightarrow)$.
Then it holds that
$(\textit{sp}_1,\sigma_1),\ldots,(\textit{sp}_n,\allowbreak \sigma_n)\in\Psi_\mathcal{Z}\Leftrightarrow
s_0\xtwoheadrightarrow{d_1}\xtwoheadrightarrow{\sigma_1}\ldots\xtwoheadrightarrow{d_n}\xtwoheadrightarrow{\sigma_n}$, 
where $d_i\in\textit{sp}_i,1\leq i\leq n$.
\end{lemma}

\begin{proof}
Let $\mathcal{A}=(L,\ell_0,\Sigma_I,\Sigma_O,\rightarrow,I)$ be a TIOA with TIOLTS $(S,s_0,\Sigma_I,\Sigma_O,\twoheadrightarrow)$ and IOLZG $(\mathcal{Z},z_0,\Sigma_I,\Sigma_O,\allowbreak\rightsquigarrow)$.
In order to prove $(\textit{sp}_1,\sigma_1),\ldots,(\textit{sp}_n,\allowbreak \sigma_n)\in\Psi_\mathcal{Z}\Leftrightarrow s_0\xtwoheadrightarrow{d_1}\xtwoheadrightarrow{\sigma_1}\ldots\xtwoheadrightarrow{d_n}\xtwoheadrightarrow{\sigma_n}$ with $d_i\in\textit{sp}_i,1\leq i\leq n$, we prove (1) $(\textit{sp}_1,\sigma_1),\ldots,(\textit{sp}_n,\allowbreak \sigma_n)\in\Psi_\mathcal{Z}\Rightarrow s_0\xtwoheadrightarrow{d_1}\xtwoheadrightarrow{\sigma_1}\ldots\xtwoheadrightarrow{d_n}\xtwoheadrightarrow{\sigma_n}$ and (2) $s_0\xtwoheadrightarrow{d_1}\xtwoheadrightarrow{\sigma_1}\ldots\xtwoheadrightarrow{d_n}\xtwoheadrightarrow{\sigma_n}$ with $d_i\in\textit{sp}_i,1\leq i\leq n\Rightarrow(\textit{sp}_1,\sigma_1),\ldots,(\textit{sp}_n,\allowbreak \sigma_n)\in\Psi_\mathcal{Z}$ separately.

\medskip

(1) It holds by Def.~\ref{def:span-trace} that $(\textit{sp}_1,\sigma_1),\ldots,(\textit{sp}_n,\sigma_n)\in\Psi_\mathcal{Z}\Leftrightarrow z_0\xrightsquigarrow{(d_1,\sigma_1)\cdots(d_n,\sigma_n)}$ with $d_i\in\textit{sp}_i,1\leq i\leq n$.
Furthermore, $\angles{\ell,D}\xrightsquigarrow{d}\angles{\ell',R(D\land g)\land I(\ell')}$ if $\exists u\in D:u\in g\land([R\mapsto0](u+d))\in R(D\land g)\land I(\ell')$ and $\angles{\ell,D}\xrightsquigarrow{(d,\sigma)}$ if $\exists\angles{\ell'',D''}\in\mathcal{Z}:\angles{\ell,D}\xrightsquigarrow{d}\angles{\ell'',D''}\xrightsquigarrow{\sigma}\angles{\ell',D'}$ by Def.~\ref{def:labeled-zone-graph}.
Here, it directly follow that $(\textit{sp}_1,\sigma_1),\ldots,(\textit{sp}_n,\allowbreak \sigma_n)\in\Psi_\mathcal{Z}\Rightarrow s_0\xtwoheadrightarrow{d_1}\xtwoheadrightarrow{\sigma_1}\ldots\xtwoheadrightarrow{d_n}\xtwoheadrightarrow{\sigma_n}$.

\medskip

(2) From Def.~\ref{def:labeled-zone-graph} and Theorem~\ref{theorem:labeled-zone-graph-operational-semantics} it follows that for all $s_0\xtwoheadrightarrow{d_1}\xtwoheadrightarrow{\sigma_1}\ldots\xtwoheadrightarrow{d_n}\xtwoheadrightarrow{\sigma_n}$ there exists a $z_0\xrightsquigarrow{(d_1,\sigma_1)\cdots(d_n,\sigma_n)}$.
Additionally, it holds by Def.~\ref{def:span-trace} that $(\textit{sp}_1,\sigma_1),\ldots,(\textit{sp}_n,\allowbreak\sigma_n)\in\Psi_\mathcal{Z}\Leftrightarrow z_0\xrightsquigarrow{(d_1,\sigma_1)\cdots(d_n,\sigma_n)}$ with $d_i\in\textit{sp}_i,1\leq i\leq n$.
Here, it directly follows that $s_0\xtwoheadrightarrow{d_1}\xtwoheadrightarrow{\sigma_1}\ldots\xtwoheadrightarrow{d_n}\xtwoheadrightarrow{\sigma_n}$ with $d_i\in\textit{sp}_i,1\leq i\leq n\Rightarrow(\textit{sp}_1,\sigma_1),\ldots,(\textit{sp}_n,\allowbreak \sigma_n)\in\Psi_\mathcal{Z}$.

\medskip

Hence, it holds that $(\textit{sp}_1,\sigma_1),\ldots,(\textit{sp}_n,\allowbreak \sigma_n)\in\Psi_\mathcal{Z}\Leftrightarrow s_0\xtwoheadrightarrow{d_1}\xtwoheadrightarrow{\sigma_1}\ldots\xtwoheadrightarrow{d_n}\xtwoheadrightarrow{\sigma_n}$ with $d_i\in\textit{sp}_i,1\leq i\leq n$.
\end{proof}

Based on this result, we are able to lift \emph{ltioco} from TIOLTS (see Def.~\ref{def:ltioco-tiolts})
to the level of IOLZG and span traces.
First, defining the two different notions of quiescence
on symbolic states of IOLZG is straightforward.
In contrast, the $\afterz$ set has now to be redefined in a recursive manner 
to consecutively traverse span traces $\xi$ instead of timed traces.
In particular, the set of symbolic states $\angles{\ell,D}$ reachable after $\xi$ is given as
the set of symbolic states reachable by all possible sequences
of timed steps comprised in $\xi$.
In a similar way, the set of \emph{suspension span traces} (\emph{sptraces})
can be defined for a symbolic state $\angles{\ell,D}$ of an IOLZG as the 
least set of span traces comprising all possible timed traces. 
Those traces are additionally equipped by
special quiescence output symbols $\delta_E$ and $\delta_S$ 
to mark occurrences of (enforced or safe) suspension.
Thereupon, the $\outz$ set can be defined as the set of all
output behaviors (\ie{} pairs $(\textit{sp},o)$ of spans \emph{sp} and output actions $o$ including quiescence) being
enabled in all symbolic states reachable from state $\angles{\ell,D}$
via span trace $\xi$ such that $\xi\cdot(\textit{sp},o)$, again, forms a valid span trace.
We further define the set $\outset(\mathcal{Z}',\xi)$ to contain the $\outz$ sets
reachable from sets $\mathcal{Z}'$ of symbolic states via span trace $\xi$.
In case of multiple output behaviors (\eg{} $(\textit{sp},o)$ and $(\textit{sp}',o)$)
with equal output actions $o$, but different spans
$\textit{sp}$, $\textit{sp}'$, we implicitly unify overlapping spans 
by requiring the set $\outset(\mathcal{Z}',\xi)$ to be minimal.
Finally, we are able to define $\mathltiocozg$ almost in the usual way, where
$\subsetsim$ is used instead of $\subseteq$ to state that all output
behaviors (\ie{} sets \textit{spa} of pairs $(\textit{sp},o)$ of spans and output actions) 
of the implementation are subsumed by those of the specification.

\begin{definition}\label{def:tiocoz-auxiliary}
Let $\textit{sp}$, $\textit{im}$ be IOLZG over $\Sigma=\Sigma_I\cup\Sigma_O$, $\gamma\in\{S,E\}$, $\angles{\ell,D}\in \mathcal{Z}$, $\mathcal{Z}'\subseteq\mathcal{Z}$, and $\xi\in((\mathbb{T}_C\times(\mathbb{T}_C\cup\{\infty\}))\times(\Sigma\cup\{\delta_\gamma\}))$.
\begin{itemize}
    \item $\angles{\ell,D}$ is \emph{safe-quiescent}, denoted by $\delta_S(\angles{\ell,D})$, iff $\forall d\in D:\angles{\ell,D}\xrightsquigarrow{d}$,
    \item $\angles{\ell,D}$ is \emph{enforced-quiescent}, denoted by $\delta_E(\angles{\ell,D})$, iff $\forall\mu\in\Sigma_O,\forall d\in D:\angles{\ell,D}\not\xrightsquigarrow{(d,\mu)}$,
    \item $(\angles{\ell,D}\afterz\xi)\subseteq\mathcal{Z}$ is the greatest set satisfying the following rules:
    \begin{itemize}
        \item $\angles{\ell,D}\in(\angles{\ell,D}\afterz\epsilon)$ and
        \item $\angles{\ell,D}\in(\angles{\ell',D'}\afterz(\textit{sp},a)\cdot\xi'')$ if $\exists d\in\textit{sp}:\angles{\ell',D'}\xrightsquigarrow{(d,a)}\angles{\ell'',D''}\land\angles{\ell,D}\in(\angles{\ell'',D''}\afterz\xi'')$,
    \end{itemize}
    \item $\sptraces(\angles{\ell,D})$ is the least set \st{} 
    $(\textit{sp}_1,\sigma_1),\ldots,(\textit{sp}_n,\sigma_n)\in\sptraces(\angles{\ell,D})\Leftrightarrow \angles{\ell,D}\xrightsquigarrow{(d_1,\sigma_1)\cdots(d_n,\sigma_n)}$ where $d_i\in\textit{sp}_i$, $1\leq i\leq n$, and $\forall\angles{\ell',D'}\in\mathcal{Z}:\angles{\ell',D'}\xrightsquigarrow{\delta_\gamma}\angles{\ell',D'}$ iff $\delta_\gamma(\angles{\ell',D'})$,
    \item $\outz(\angles{\ell,D}, \xi)\subseteq(\mathbb{T}_C\times(\mathbb{T}_C\cup\{\infty\})\times(\Sigma_I\cup\Sigma_O\cup\{\delta_\gamma\}))$ is the greatest set \st{}
    $(\textit{sp},o)\in\outz(\angles{\ell,D},\xi)$ if $\angles{\ell,D}\xRightsquigarrow{o}\,\land\,\xi\cdot(\textit{sp},o)\in\sptraces(z_0)\land o\in\Sigma_O\cup\{\delta_\gamma\}$,
    \item $\outset(\mathcal{Z}',\xi)$ is the least set \st{} $\forall(\textit{sp},o)\in\bigcup_{z\in\mathcal{Z}'}\outz(z,\xi):\exists(\textit{sp}',o)\in\outset(\mathcal{Z}',\xi):\textit{sp}\preceq\textit{sp}'$,
    \item $\textit{im}\mathltiocozg \textit{sp}:\Leftrightarrow\forall\xi\in\sptraces(\textit{sp}):\outset(\textit{im}\afterz\xi,\xi)\subsetsim\outset(\textit{sp}\afterz\allowbreak\xi,\xi)$, where $\textit{spa}\subsetsim\textit{spa}'\Leftrightarrow\forall(\textit{sp},o)\in\textit{spa}:\exists(\textit{sp}',o)\in\textit{spa}':\textit{sp}\preceq\textit{sp}'$
    \end{itemize}
\end{definition}

\begin{example}
Considering the running example in Figs.~\ref{fig:vending-machine-ta} and~\ref{fig:vending-machine-tioco}, we observe that $\llbracket\mathcal{A}_1\rrbracket\mathltiocozg\llbracket\mathcal{A}_1'\rrbracket$ does not hold.
Let $\xi=((20,\infty),\textit{?press}),((0,20),\allowbreak\textit{?press})$.
Then $((0,\infty),\delta_S)\in\outset(\llbracket A_1\rrbracket_\mathcal{Z}\afterz\xi,\xi)$ and $((0,\infty),\delta_S)\notin\outset\allowbreak(\llbracket A_1'\rrbracket_\mathcal{Z}\afterz\xi,\allowbreak\xi)$ as it is not safe to wait in \emph{add sugar} of $\mathcal{A}_1'$ due to the invariant $y\leq15$.
\end{example}

Finally, we prove that for any two TIOA $\mathcal{A}_\textit{im}$ and 
$\mathcal{A}_\textit{sp}$, checking $\mathltiocozg$ on IOLZG is equivalent 
to checking $\mathltiocolts$ on TIOLTS.

\begin{theorem}[Correctness of $\mathltiocozg$]\label{theorem:soundness-zg}
    Let $\mathcal{A}_\textit{im}$ and $\mathcal{A}_\textit{sp}$ be TIOA.
    $$\llbracket\mathcal{A}_\textit{im}\rrbracket_{\mathcal{Z}}\,\mathltiocozg\,\llbracket\mathcal{A}_\textit{sp}\rrbracket_{\mathcal{Z}}\,\Leftrightarrow\,
      \llbracket\mathcal{A}_\textit{im}\rrbracket_{S}\,\mathltiocolts\,\llbracket\mathcal{A}_\textit{sp}\rrbracket_{S}$$
\end{theorem}

\begin{proof}
Let $\mathcal{A}_\textit{im}$ and $\mathcal{A}_\textit{sp}$ be TIOA.
Lemma~\ref{lemma:span-trace-correctness} shows that $(\textit{sp}_1,\sigma_1),\ldots,(\textit{sp}_n,\sigma_n)\in\Psi_\mathcal{Z}\Leftrightarrow s_0\xtwoheadrightarrow{d_1}\xtwoheadrightarrow{\sigma_1}\ldots\xtwoheadrightarrow{d_n}\xtwoheadrightarrow{\sigma_n}$ with $d_i\in\textit{sp}_i,1\leq i\leq n$.
Hence, $(\textit{sp}_1,\sigma_1),\ldots,(\textit{sp}_n,\sigma_n)\in\sptraces(\llbracket\mathcal{A}_\textit{sp}\rrbracket_\mathcal{Z})\Leftrightarrow(d_1,a_1),\ldots(d_n,a_n)\in\tstracesl(\llbracket\mathcal{A}_\textit{sp}\rrbracket_S)$ with $d_i\in\textit{sp}_i,1\leq i\leq n$ as applying symbols $\delta_S$ and $\delta_E$ to the sets of $\tstracesl$ and $\sptraces$ is done in the same manner (\cf{}~Defs.~\ref{def:ltioco-tiolts} and~\ref{def:tiocoz-auxiliary}).
It remains to be shown that $(d,o)\in\outl(\llbracket\mathcal{A}_\textit{sp}\rrbracket_S\allowbreak\after(d_1,a_1),\ldots,(d_n,a_n))\Leftrightarrow(\textit{sp},o)\in\outset(\llbracket\mathcal{A}_\textit{sp}\rrbracket_\mathcal{Z}\afterz(\textit{sp}_1,\allowbreak a_1),\ldots,(\textit{sp}_n,a_n))$ with $d\in\textit{sp}$ and $d_i\in\textit{sp}_i,1\leq i\leq n$.
This directly follows from the first part of this proof as, by definition, $(d,o)\in\outl(\llbracket\mathcal{A}_\textit{sp}\rrbracket_S\allowbreak\after(d_1,a_1),\ldots,\allowbreak(d_n,a_n))\Leftrightarrow(d_1,a_1),\ldots,(d_n,a_n),(d,o)\in\tstracesl$.
Hence, it holds that $\llbracket\mathcal{A}_\textit{im}\rrbracket_{\mathcal{Z}}\allowbreak\,\mathltiocozg\allowbreak\llbracket\mathcal{A}_\textit{sp}\rrbracket_{\mathcal{Z}}\,\Leftrightarrow\, \llbracket\mathcal{A}_\textit{im}\rrbracket_{S}\,\mathltiocolts\,\llbracket\mathcal{A}_\textit{sp}\rrbracket_{S}$.
\end{proof}

From Theorems~\ref{theorem:soundness-ltioco} and~\ref{theorem:soundness-zg} 
it also follows that $\mathltiocozg$ is sound with respect to $\mathtiocodelta$ 
and from Theorems~\ref{lemma:ltioco-preorder} and~\ref{theorem:soundness-zg} it 
follows that $\mathltiocozg$ is a preorder on input-enabled IOLZG.
Finally, we can likewise conclude compositionality of $\mathltiocozg$.

\begin{corollary}\label{corollary:compositionality-zone-graph}
Let $\textit{im}_{1}$ and $\textit{im}_{2}$ as well as $\textit{sp}_{1}$ and $\textit{sp}_{2}$ 
be input-enabled and composable TIOA enabling independent progress.
Then $(\llbracket\textit{im}_1\rrbracket_{\mathcal{Z}} \mathltiocozg \llbracket\textit{sp}_1\rrbracket_{\mathcal{Z}}) \land
(\llbracket\textit{im}_2\rrbracket_{\mathcal{Z}} \mathltiocozg \llbracket\textit{sp}_2\rrbracket_{\mathcal{Z}}) \Rightarrow
\llbracket\textit{im}_1\parallel\textit{im}_2\rrbracket_{\mathcal{Z}} \mathltiocozg\allowbreak \llbracket\textit{sp}_1\parallel\textit{sp}_2\rrbracket_{\mathcal{Z}}$.
\end{corollary}

\section{Tool Support}\label{sec:implementation}
To show practical feasibility of our technique, we implemented a 
tool based on the concepts of the \textsc{JTorX} tool~\cite{Belinfante2010,Tretmans2003}, 
originally being developed for (untimed) \textbf{ioco} testing.
Similar to \textsc{JTorX}, our tool supports \emph{online white-box testing}: a running implementation is 
investigated on-the-fly whether it is conforming to a specification both given as TIOA.
Our tool supports a generic interface enabling it to be used
for checking any kind of implementation 
(in the current version, the interface is implemented to accept TIOA models as implementation).
To check conformance of a given implementation to a specification, the tool checks $\mathltiocozg$ 
on the labeled zone-graph representations of both TIOA models.
As input TIOA models, our tools supports the exchange format of 
\textsc{Uppaal}~\cite{Larsen1997} (a mature model checker for timed systems).

Internally, our tool uses \emph{Difference Bound Matrices} (DBM) 
being an efficient representation of zones~\cite{Bellman1957,Dill1990,Bengtsson2004}.
In particular, DBM-based representations of zones provide
comparison operators ${\sim}\in\{<,\leq,=,\geq,>\}$.
For a consistent representation, a fresh clock $0_C$ (with constant value zero) 
is introduced resulting in the set of clocks $\mathcal{C}_0=\mathcal{C}\cup\{0_C\}$ 
in which each clock is aligned to $0_C$.
Based on this construction, atomic clock constraints of the form $x\sim r$ 
can be represented as $x-y\preceq r$ with ${\preceq}\in\{<,\leq\}$.
Hence, every zone $D\in\mathcal{C}_0$ can be represented
with a maximum of $|\mathcal{C}_0|^2$ atomic clock constraints, 
and therefore, each zone may be described as a matrix 
of size $|\mathcal{C}_0|\times|\mathcal{C}_0|$~\cite{Bengtsson2004}.
Each entry $D_{i,j}$ (row $i$, column $j$)
thus refers to the atomic clock constraint $x_i-x_j\preceq r$.
Hence, entries of the matrix are pairs of difference 
values $r_{i,j}$ and comparison operators in $\preceq$, being derived as follows.
For every entry $D_{i,j}$, we set the value $r_{i,j}$ such that $x_i-x_j\preceq r_{i,j}$ holds.
If a difference is unbounded (\ie{} $x_i$ and $x_j$ are not related by any constraint), 
we set the value to $r_{i,j}=\infty$.
Additionally, we have to require clocks to have
non-negative values (\ie{} $0_C-x_i\leq0$).

\begin{example}
Figure~\ref{fig:dbm-example} depicts the DBM for the zone $\angles{1\leq x\leq 2,y\leq 2}$.
For instance, $D_{x,y}=D_{y,x}=\infty$ as $x$ and $y$ are not related by a comparison.
Additionally, $D_{0_C,x}=(-1,\leq)$ as $1\leq x\leq2$ such that $0-x\leq-1$.
Furthermore, $D_{x,0_C}=2$ due to $x\leq2$.
\begin{figure}[tp]
	\centering
	\scalebox{\scalefactor}{
	\footnotesize
	$
	\bordermatrix{
		    & 0_C      & x         & y \cr
		0_C & (0,\leq) & (-1,\leq) & (0,\leq) \cr
		x   & (2,\leq) & (0,\leq)  & \infty \cr
		y   & (2,\leq) & \infty    & (0,\leq) \cr
	}
	$
	}
	\caption{Difference Bound Matrix for the Zone $\angles{1\leq x\leq 2,y\leq 2}$}\label{fig:dbm-example}
\end{figure}
\end{example}
The tool is available online at \url{https://www.es.tu-darmstadt.de/ltioco}.
%
\section{Related Work}\label{sec:related-work}
Several versions of \textbf{tioco} have recently been
proposed~\cite{Schmaltz2008,BrandanBriones2004,Hessel2008,Krichen2004,Larsen2004}, 
whereas \textbf{ltioco} is, to the best of our knowledge, the first approach
working on the symbolic thus finite zone-graph representation of TIOA instead of 
infinitely branching TIOLTS.
The only other existing symbolic variant 
of \textbf{tioco} is based on \emph{symbolic} timed automata 
with data variables, but does neither include quiescence 
nor ensure finiteness of the state space~\cite{Styp2010}.
In addition, our novel notions of timed quiescence are different
from any existing approach, where absence of outputs is either considered
only up to a fixed bound $M$~\cite{BrandanBriones2004,Hessel2008}, or 
for all possible delays~\cite{Schmaltz2008,Larsen2004}.
Recent tools implementing variants of \textbf{tioco}~\cite{Larsen2005,Bohnenkamp2005,Krichen2004} 
also mostly differ in their interpretation of quiescence which
can all be simulated in our framework, but not vice versa.
Moreover, neither of these approaches distinguishes safe from enforced quiescence as done in our approach.

In addition, compositionality properties have only 
been considered in~\cite{Bannour2013} so far, where again no notion of quiescence is considered. 
Furthermore, there are techniques for test-generation from TIOA models.
In order to handle infinitely branching state spaces, En-Nouaary and Dssouli~\cite{En-Nouaary2003} 
derive test cases only for a particular subset of TIOA behaviors, whereas,
similar to our approach, \mbox{Brand\'an Briones} and R\"ohl~\cite{BrandanBriones2005} use a zone-based 
representation.
However, the latter approach is limited to deterministic TA, 
which are strictly less expressive than our TIOA.
Springintveld \etal{}~\cite{Springintveld2001} propose an algorithm 
for exhaustive black-box test generation for timed systems, but no notions
of quiescence are taken into account.

Besides adopting \textbf{ioco}-like conformance notions
to timed systems as done by the different variants of 
\textbf{tioco}, the only other timed implementation-relation theory we are aware of uses
a refinement-based implementation relation~\cite{David2010}.
Moreover, Bornot \etal{}~\cite{Bornot2000} investigate requirements
for ensuring liveness-by-construction of timed systems 
using trace-based composition operators for\linebreak TIOLTS, whereas
conformance theories are out of scope.

Finally, there are several other \textbf{ioco}-based testing theories.
Among others, \textbf{mioco}~\cite{Luthmann2019,Luthmann2016,Luthmann2015mioco} (\ie{} \textbf{ioco} for modality-based systems) distinguishes optional transition (which may be implemented) from mandatory transitions (which must be implemented).
Furthermore, \textbf{featured-ioco}~\cite{Beohar2016} is based on so-called featured transition systems, incorporating feature constraints to to restrict which (pairs of) transitions may be part of the same variant.
However, none of these approaches considers real-time constraints.

\section{Conclusion}\label{sec:conclusion}
We presented an improved version
of a timed input/output conformance testing relation, called \textbf{ltioco}, 
to ensure not only safe but also \emph{live} behaviors 
of implementations with time-critical behaviors modeled as TIOA.
Additionally, we investigated compositionality properties of \textbf{ltioco} 
and we extended the construction of zone graphs 
to check \textbf{ltioco} on a finite semantic representation of TIOA.
As a future work, we plan to enrich our framework by further operators including
quotienting and conjunction as well as refinement~\cite{David2010}
and to extend our tool implementation by automated 
test-generation and test-execution capabilities.
Furthermore, we plan to evaluate our approach by applying our tool 
to a number of well-known case studies (\eg{}~\cite{Jensen1996,Havelund1997,Lindahl2001}).

\bibliographystyle{eptcs}
\bibliography{content/references}

\end{document}